\documentclass[journal]{IEEEtai}

\usepackage[colorlinks,urlcolor=blue,linkcolor=blue,citecolor=blue]{hyperref}

\usepackage{color,array}

\usepackage{graphicx}

\usepackage{mathrsfs}
\usepackage{amsfonts}
\usepackage{amssymb}
\usepackage{graphicx}
\usepackage{amsthm}
\usepackage{amsmath}
\usepackage{color}
\usepackage{subfigure}
\usepackage{float}
\usepackage{booktabs}
\usepackage{enumerate}

\usepackage{epstopdf}
\usepackage{verbatim}
\usepackage{morefloats}
\usepackage[misc]{ifsym}
\usepackage{cite}
\usepackage{amsmath}
\usepackage{threeparttable}
\usepackage{multirow}

\usepackage[justification=centering]{caption}

\newtheorem{property}{Property}
\newtheorem{definition}{Definition}
\newtheorem{assumption}{Assumption}
\newtheorem{lemma}{Lemma}
\newtheorem{theorem}{Theorem}{}
{}
\newtheorem{remark}{Remark}{}

\begin{document}

\title{Adaptive Feedforward Neural Network Control with  an Optimized Hidden
Node Distribution} 

\author{Qiong~Liu,~\IEEEmembership{} 
        Dongyu~Li,~\IEEEmembership{}             
        Shuzhi~Sam~Ge,~\IEEEmembership{}and 
        Zhong~Ouyang~\IEEEmembership{}

 \thanks{ Qiong Liu is with Tsinghua-Berkeley Shenzhen Institute, Tsinghua University,  Shenzhen, Guangdong, 518055 China.}
  \thanks{ Dongyu Li is with the School of Cyber Science and Technology, Beihang University, Beijing
100191, China. (Corresponding author: Dongyu Li, dongyuli@buaa.edu.cn)}
\thanks{ Shuzhi Sam Ge is with Department of Electrical and Computer Engineering, National University of Singapore, Singapore 117576.}
\thanks{ Zhong Ouyang is with Department of Mechanical and Aerospace Engineering, The Ohio State University, Columbus, OH 43202, USA.} 
}


\maketitle

\begin{abstract}
	Composite adaptive radial basis function neural network (RBFNN) control with a lattice distribution of hidden nodes has three inherent demerits: 1) the approximation domain of adaptive RBFNNs is difficult to be determined a priori; 2) only a partial persistence of excitation (PE) condition can be guaranteed; and 3) in general, the required number of hidden nodes of RBFNNs is  enormous. 
	This paper proposes an adaptive feedforward RBFNN controller with an optimized distribution of hidden nodes to suitably address the above demerits. The distribution of the hidden nodes calculated by a K-means algorithm is optimally distributed along the desired state trajectory.  
	The adaptive RBFNN satisfies the PE condition for the periodic reference trajectory. The weights of all hidden nodes will converge to the optimal values. This proposed method considerably reduces the number of hidden nodes, while achieving a better approximation ability.
The proposed control scheme shares a similar rationality to that of the classical PID control in  two special cases, which can thus be seen as an enhanced PID scheme with a better approximation ability. 
	 For the controller implemented by digital devices, 
	 the proposed method, for a manipulator with unknown dynamics, potentially achieves  better control performance than model-based schemes with accurate dynamics. 
	  Simulation results demonstrate the effectiveness of  the proposed scheme. This result provides a deeper insight into the coordination of the adaptive neural network control and the deterministic learning theory.
\end{abstract}

\begin{IEEEImpStatement}
	Adaptive RBFNN control   learns to control a robot manipulator when both the structures and parameters of the target robot are unknown in advance. 		
Unfortunately, current adaptive RBFNN controllers need a large-scale neural network to approximate the dynamics of the robot manipulator, and the learning performance cannot be guaranteed to converge.
	The proposed method in this paper not only reduces the scale of neural networks to substantially  alleviate the computational burden but also evidently achieves better learning performance. 
Simulation examples show  that this method increases the control accuracy by more than $9$ times and reduces the scale of neural networks by $35$ times as compared to the traditional lattice scheme.
	Intuitively, people usually believe that a model-based controller with an accurate dynamic model may achieve the best control performance. 	However, compared with the model-based controller with an accurate dynamic model,  the proposed control scheme with an unknown dynamic model even further increases the control accuracy by  1.5 times.
	This technology provides a more straightforward path for engineers, who may be not   experts in  complicated control system analysis methods,  to design an adaptive robotic controller to achieve enhanced performance.
\end{IEEEImpStatement}

\begin{IEEEkeywords}
Adaptive neural network control, Deterministic learning,  Persistence of excitation.
\end{IEEEkeywords}

\section{Introduction}

%
%
%
%

\IEEEPARstart{A}{daptive} radial basis function neural network (RBFNN) control is an effective way to handle   uncertainties of   system dynamics when both the structures and parameters are unknown \cite{peng2019force, Hewei2018, shi_P_2020_Cybernetics,arabi2019neuroadaptive }.
  	RBFNNs with deterministic hidden nodes have a higher learning speed than both multilayer neural networks and RBFNNs with adjustable hidden nodes \cite{ge2001stable}.
The learning mechanism of the adaptive RBFNN control with   deterministic hidden nodes was named deterministic learning in \cite{CongWang2006}.
  Generally, there are two structures to accomplish adaptive RBFNN control:  composite adaptive RBFNN control and  adaptive feedforward RBFNN control.
  
   Composite adaptive RBFNN control is derived from composite adaptive control  \cite{slotine1989composite}.  
It is widely utilized in the adaptive RBFNN control community due to   superior performance and rigorous   proof. However, the following essential  issues  are to be further investigated:
  \begin{itemize}
  \item[1)]
  The inputs of the RBFNN include the desired states and the state errors, but the state errors are hard to be known a priori.  Therefore, it is difficult to determine the approximation domain of the RBFNN.
When the inputs leave the approximation domain, the outputs of the RBFNN vanish. This fails the approximation\cite{hewei7994622, ren2010adaptive, gao8879661, yang2018robot}. 
 	The sliding-mode control scheme can push the states to the domain again \cite{Sanner1992, zhao2007locally}, but it needs extra information about the   target system. 
 	The Barrier Lyapunov   function-based method  constrains the state errors in the designed intervals   \cite{tee2009barrier, liu2016barrier, Gao8811752, huang2019motor}, whereas it complicates the controller design and needs hardware with a high sampling rate.

  \item[2)] 
  
Only a partial persistence of excitation (PE) condition, i.e., the PE condition of a certain regression subvector constructed out of the RBFs along the periodic system trajectory, is proven to be satisfied in \cite{CongWang2006, Wang2015PElevel, Wang2016PElevel,  Wang2017PElevel, Wang2019_he}.  
 Therefore, only the corresponding regression subvector weights converge to their optimal values for a periodic reference trajectory. 
 For the other hidden nodes  (that do not satisfy the PE condition), their weights might not converge to their optimal values. {This substantially }degrades the robustness of the controller. 
	In addition, the partial PE condition needs a three-step procedure to be guaranteed in  \cite{CongWang2006}.

  \item[3)] The number of hidden nodes under a lattice distribution is $m^p$, where $m$ is the number of the hidden nodes in each dimension and $p$ is the dimension of the input vectors of the RBFNN.  It grows exponentially with respect to the dimension of the input vectors and grows polynomially with respect to the number of the hidden nodes in each dimension.
	The dimension of the input vectors is determined by both the controller structure and the degree of freedom (DOF) of the controlled system. In addition, to achieve better approximation performance, we need more hidden nodes in each channel.
	 For high DOF robot manipulators, as the values of $m$ and $q$ become large, the number of the hidden nodes becomes unacceptably huge, hence  inevitably limiting the application in  practical electrical devices \cite{chen_wang2019}.
   
  
  \end{itemize}
  
  Adaptive feedforward RBFNN control is derived from PD-plus-feedforward control \cite{slotine1987on}.
 PD-plus-feedforward control has a simple structure, and the inputs of its  feedforward term are only the desired state.
However, compared with the composite adaptive control structure, the PD-plus-feedforward control structure is seldom utilized in adaptive RBFNN controllers. The possible reasons may be as follows:
  the control gains should be large enough to suppress the residual error between the composite dynamics and the feedforward dynamics, whereas the explicit value of control gains cannot be determined for unknown dynamics; this undermines the rigorousness of stability analysis. 
  However, from practical applications in robotic manipulators, PD-plus-feedforward control has similar control performance to computed torque control \cite{Chae1987}  or even better control performance under background noise and imprecise system dynamics
 \cite{khosla1988experimental,  reyes2001experimental}. The adaptive feedforward RBFNN control with a lattice distribution of hidden nodes proposed in \cite{chen2012globally, pan2016hybrid, pan2016biomimetic} partially solves the aforementioned issues as follows:

\begin{itemize}
\item[1)]  The inputs of the adaptive RBFNN are the desired trajectories of the robotic manipulator, known a priori. Based on the specific inputs, the approximation domain can be properly determined.

\item[2)] The periodic states of the closed-loop system in step (3) of \cite{CongWang2006} are no longer required. The PE condition is guaranteed beforehand  by utilizing the desired periodic states as the inputs of adaptive RBFNNs.

\item[3)] The dimension of inputs is reduced to $3n$, where $n$ is the DOF of the robotic manipulator; this simplifies the control structure and reduces the number of hidden nodes. 
\end{itemize}


	Inspired by the above literature, we propose an adaptive feedforward RBFNN control scheme with an  optimized distribution of hidden nodes. 
	The controller includes a PD term and an adaptive feedforward RBFNN term. 
	The position of the hidden nodes calculated by the K-means algorithm is optimally distributed along the desired state trajectory which is essentially different from the lattice distribution.  
	Compared with  the existing adaptive RBFNN control methods, this paper has the following improvements on the remaining issues mentioned above while also considering other aspects:

\begin{itemize}


\item[1)] This paper proposes an adaptive feedforward RBFNN controller with an optimized distribution of hidden nodes. The PE condition, rather than the partial PE \cite{CongWang2006,   Wang2017PElevel, Wang2016PElevel, Wang2015PElevel, pan2016biomimetic, Wang2019_he} of the adaptive RBFNN controller, is guaranteed before the control process. This brings an attractive advantage that all weights of the RBFNN would converge to their optimal values.
	
\item[2)] 
Compared with the adaptive RBFNN control with lattice hidden nodes \cite{Sanner1992,  CongWang2006, sun2018fuzzy, zeng2014learning}, our proposed scheme reduces the number of hidden nodes significantly whereas  it has better approximation and control performance. The number of hidden nodes are 1403 in \cite{Sanner1992}, 441 in \cite{CongWang2006},  $2^8$  in \cite{sun2018fuzzy}, and $4^6$ in \cite{zeng2014learning},  but our scheme only uses $20$ hidden nodes.
 
\item[3)] We find two unique relations between the PID control and the adaptive feedforward RBFNN control with an optimized distribution of hidden nodes:  i. when the width of the RBFNN approaches infinity,  the controller degrades to a PID controller; and ii. for the set-point tracking problem, the adaptive feedforward RBFNN controller with an optimized distribution of hidden nodes is the same as the PID controller. Thus, the PID controller can be seen as a special case of  the adaptive feedforward RBFNN control with an optimized distribution of hidden nodes. Few works of literature explain the adaptive RBFNN control from such perspectives.


\item[4)] For controllers implemented on digital devices, our proposed scheme with the unknown dynamics has the potential to achieve  better control performance than the model-based control schemes with accurate dynamics, which has been fully shown in simulation results.
  \end{itemize}

    The rest of this paper is organized as follows.
Section \ref{Problem_Formulation} presents the problem formulation and preliminaries. 
     The main results are given in Sections \ref{control_design}. In Section  \ref{Discussions}, we explain the rationality that the controller can be treated as an enhanced PID controller.
      Simulation examples are presented in Section \ref{Simulation}. Finally, conclusions are given in Section     \ref{Conclusion}.  
     In this paper, $\| \cdot  \|$  stands for Euclidean norm of vectors or induced norm of matrices.  $\mathbb{R}$, $\mathbb{R}^+$, $\mathbb{R}^n$, $\mathbb{R}^{n\times m}$ represent  the set of real numbers, positive real numbers,   real $n$-vectors,  real $n \times m$  matrices, respectively. \( \lambda_{\max }(\cdot) \) and \( \lambda_{\min }(  \cdot ) \) denote the largest and smallest  eigenvalues of the corresponding square matrix. 
     $\min( \cdot)$ and $\max( \cdot)$ are the minimum and maximum values of the corresponding function. $\mathcal{C}^{k}$ represents the $k-$order derivative of a continuous function. $L_\infty$ denotes the space of bounded signals. $\Omega_{x}:=\{x| \|x\| \leq \bar{x} \}$ is a ball of radius $\bar{x}$ with $\bar{x}$ being a positive constant.

\section{Problem Formulation and Preliminaries}  \label{Problem_Formulation}

Here, we first introduce the dynamics of robot manipulators in Section \ref{Dynamics_Descriptions}, and then an RBFNN is introduced to approximate the unknown dynamics of a manipulator in Section \ref{Function_Approximation}. In Section \ref{K-means_Algorithm},  the K-means algorithm is introduced to improve the approximation ability of an RBFNN by optimally distributing its hidden nodes along the desired state trajectory, and then the PE condition for adaptive RBFNN control with an optimized distribution of hidden nodes is studied in Section \ref{PE_4}.

\subsection{Dynamics Descriptions} \label{Dynamics_Descriptions}
Consider the class of robot manipulators described as follows \cite{ortega2013passivity}:
\begin{equation}\label{system}
M(q)\ddot{q}+C(q,\dot{q})\dot{q}+G(q)=\tau,
\end{equation}
where $q,\dot{q},\ddot{q}\in \mathbb{R}^n$ are  the vectors of the joint positions, joint velocities and joint accelerations, respectively, $\tau\in \mathbb{R}^n$ is the input torque vector, $M(q)\in \mathbb{R}^{n\times n}$ is the inertia matrix, $C(q,\dot{q})\in \mathbb{R}^{n\times n}$ is the Coriolis matrix, $G(q)\in \mathbb{R}^n$ is the gravity force and $n$ is the  DOF of the system. In addition, $q$ and $\dot{q}$ are measurable and the dynamics is unknown.


\begin{property} \label{bounded_M}
The matrix $M(q)$ is a symmetric and positive definite matrix and satisfies $\lambda_m I \leq M(q) \leq \lambda_MI$, where $\lambda_m$ and $\lambda_M$ are the minimum and maximum eigenvalues of $M(q)$.
\end{property}

\begin{property}\label{skew}
The matrix $\dot{M}(q)-2C(q,\dot{q})$ is skew-symmetric, and $z^T\big(\dot{M}(q)-2C(q,\dot{q})\big)z=0$, $\forall z \in \mathbb{R}^n$.
\end{property}


\begin{assumption}\label{periodic}
The desired state trajectory $Z_d= [q_d^T, \dot{q}_d^T,\ddot{q}_d^T ]^T$ is periodic continuous and bounded such that  $\|Z_d\| \leq \bar{Z}_d$ with $ \bar{Z}_d\in \mathbb{R}^+$ being a positive constant.
\end{assumption}


\subsection{Function Approximation}\label{Function_Approximation}
The activation functions of adaptive RBFNNs are the Gaussian function taking the following form:
\begin{equation}
S_{j}(Z)=\exp\big[-\frac{(Z-\mu_j)^T(Z-\mu_j)}{\sigma^2} \big],  j=1,2,\dots,m,
\end{equation} 
where $\mu_j=[\mu_{j1},\mu_{j2},\dots, \mu_{jp}]^T \in \mathbb{R}^q$ is the position of the hidden node,  $Z=[Z_1,Z_2,\dots, Z_p]^T\in \Omega_{Z} \subset  \mathbb{R}^q$ is the input vector,    $\sigma$ is the width of the  Gaussian function, and $\Omega_{Z}$ is the approximation domain of RBFNNs.

The function can be approximated as:   
\begin{equation} \label{three RBFNN}
F(Z)=W^{*T}S(Z)+\epsilon(Z),\ 
  \forall Z\in \Omega_Z,
\end{equation}
where  $F(Z)\in \mathbb{R}^n$, $W^*=[W_1^*, W_2^*, \dots, W_n^*]\in \mathbb{R}^{m \times n}$, and $S(Z)\in \mathbb{R}^{m}$.

	The ideal weights of the RBFNNs are defined as: 
	\begin{equation}
W^*:=\arg \min \limits_{W}\{\sup\limits_{Z\in \Omega_Z} |F-\hat{W}^TS(Z)|\}.
\end{equation}

%

The estimated error of $W$ is:
\begin{equation}
\tilde{W}=W^* - \hat{W}.
\end{equation}

\begin{remark}
In RBFNNs, there are three classes of parameters, including $\mu_j$, $\sigma$, and $\hat{W}$. For linearly parametrized neural networks, the structure information contains the values of  $\mu_j$ and $\sigma$, which need to be determined before the control process. $\hat{W}$ is adapted in the training process.  However, in non-linearly parametrized neural networks,  $\mu_j$, $\sigma$, and $\hat{W}$ are adapted in the training process.   Compared with non-linearly parametrized neural networks, the deterministic structure of linearly parametrized neural networks  brings a higher learning speed.
\end{remark}

\subsection{K-means Algorithm} \label{K-means_Algorithm}
K-means is an unsupervised learning algorithm to solve clustering problems.
 The method aims at minimizing the objective function:
\begin{equation}
J(\mu)=\sum_{j=1}^{m} \sum_{i=1}^{n}\left\|Z_{i}^{(j)}-\mu_{j}\right\|^{2},
\end{equation}
where \( \left\|Z_{i}^{(j)}-\mu_{j}\right\|^{2} \) is a  squared Euclidean distance between a data point \( Z_{i}^{(j)} \) and the corresponding cluster center \( \mu_{j} \).

The K-means algorithm proceeds in the following four steps:
\begin{itemize}
\item[1.] Randomly choose the initial $m$ centres closing to zeros $\mathcal{\mu}=\left\{\mu_{1}, \mu_{2}, \cdots, \mu_{m}\right\}$.
\item[2.] Compute the distance from  each data point $Z_{i}$ to each center $\mu_{j}$.   Choosing the center $\mu_{j}$ which is closest to data point $Z_{i}$ and assign the data point $Z_{i}$ to the cluster which is represented by $Z_{i}^{(j)}$.
\item[3.] Update the center by the means of the data points for each cluster.
\item[4.] Repeat Steps 2 and 3 until $\mu$ convergence.
\end{itemize}


K-means is applicable only when the input data is obtained before the learning process; this limits its implementation on the composite adaptive RBFNN control method because the state error in this control method is difficult to  know beforehand. Nonetheless, the K-means algorithm can be implemented on the adaptive feedforward RBFNN control scheme because the input is the desired state of robotic manipulators known as   a priori.


An example of the distribution of hidden nodes of the RBFNN for one link of a robot manipulator is given in Fig. \ref{hidden_nodes},  where we can see that the hidden nodes are optimally distributed along the reference trajectory. In this example, the reference trajectory is $Z_d = [q_d, \dot{q}_d, \ddot{q}_d]^T = [\sin(t), \cos(t), -\sin(t)]^T$ and the number of hidden nodes of the RBFNN is chosen as 20.

\begin{figure}[!t]
\centering 
\includegraphics[width=3 in]{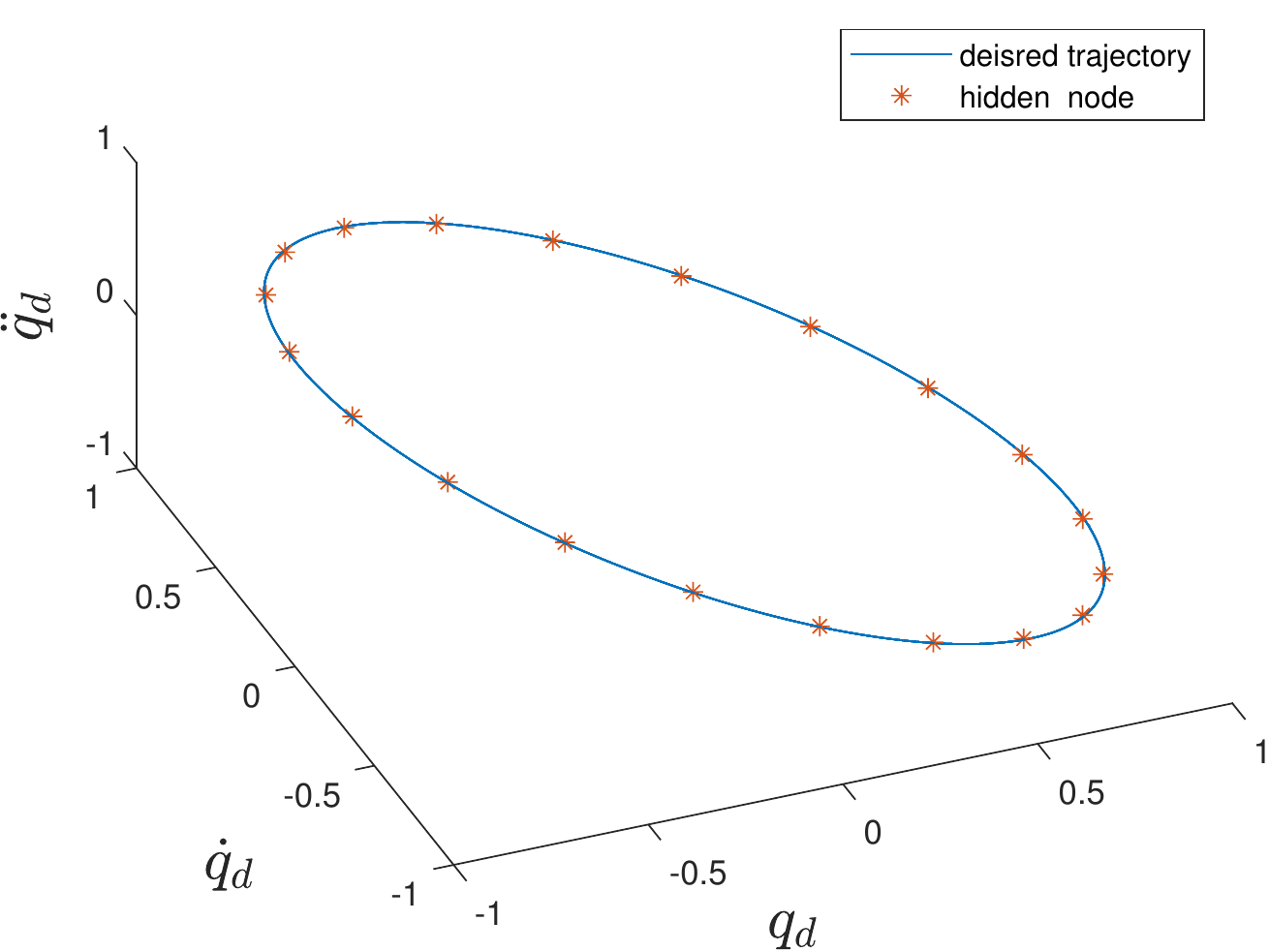}
\caption{The distribution of hidden nodes and the desired trajectory for joint 1.}
\label{hidden_nodes}
\end{figure}

	Compared with the lattice distribution of hidden nodes scheme, the  optimized  distribution of hidden nodes scheme improves the approximation ability of RBFNNs significantly.
	For  the specific approximation error, the  optimized distribution of hidden nodes scheme will sharply decrease the number of the hidden nodes and reduce the complicity of the RBFNN structure. 	
	A simpler RBFNN structure will significantly alleviate the implementation burden in terms of the hardware selection, algorithm realization, and system debugging. 
	One can comprehend the advantages of the optimized distribution of hidden nodes scheme from the following two aspects: 1) all of the hidden nodes are effective for the approximation task; the hidden nodes useless for the approximation task disappear, thus reducing the number of hidden nodes  significantly, especially for RBFNNs with high input dimensions; and 2) each hidden node works under its maximum approximation capability.

\subsection{The PE Condition for Adaptive RBFNN Control with an  Optimized Distribution of Hidden Nodes} \label{PE_4}
	The PE condition is an  essential  property in adaptive control and system identification. It determines whether the adaptive parameters would converge to their optimal or exact values.

\begin{definition}\cite{CongWang2006}: \label{definition_PE}
A piecewise-continuous, uniformly-bounded, vector-valued function \( S:[0, \infty) \rightarrow  \mathbb{R}^{m} \) is said to satisfy the PE condition, if there exist positive constants \( \alpha_{1}, \alpha_{2}, \) and \( T_{0} \) such that:
\[
\alpha_{1} I \geq \int_{t_{0}}^{t_{0}+T_{0}} S(\tau) S(\tau)^{T} d \tau \geq \alpha_{2} I \quad \forall t_{0} \geq 0,
\]
where \( I \in \mathbb{R}^{m \times m} \) is the identity matrix,  $\alpha_{1}$ and $\alpha_{2}$ in (9) are the level
of excitation and the upper bound of excitation.
\end{definition}

According to this definition, the PE condition requires that the integral of the matrix $S(\tau) S(\tau)^{T}$ should be uniformly positive definite throughout the length $T_0$.
 It is noted that if $S(\tau)$ satisfies the PE condition for the time interval $[t_0,t_0+T_0]$, it satisfies the PE condition for any interval $[t_0,t_0+T_1]$ when  $T_1>T_0$ \cite{Kurdila1995}.

\begin{lemma}\label{lemma_PE}
 Consider an arbitrary continuous periodic trajectory \( Z_d(t) \) \( : \mathbb{R}^{+} \mapsto \Omega_{Z_d} \) with period \( t_{p} \). When the hidden nodes of the RBFNN are  optimally distributed along the trajectory $Z_d(t)$, $ Z_d(t)$ can visit each hidden nodes in any period $[t_0,t_0+t_p]$. Then, for an arbitrary time $T_0 \geq 2t_p$, the regressor \( S(Z_d) \) satisfies the PE condition. 
\end{lemma}

 The proof can be given by following that of Lemma 1 in \cite{lu1998robust}.

\begin{lemma}\label{PE_level} \cite{zheng2017relationship}
The PE level increases with the increasing separation distance of hidden nodes of RBFNNs. The approximation error $\epsilon(Z)$ increases with the increasing separation distance and  the decreasing fill distance of hidden nodes of RBFNNs. The fill distance is defined as $
h_{\mu, \Omega_{Z}}:=\sup \limits_{Z \in \Omega_{Z}} \min _{\mu_{j} \in \mu}\left\|Z-\mu_{j}\right\|$. The convergence rate of the deterministic learning increases with the PE level. The convergence accuracy of the deterministic learning increase with both the PE level and the approximation accuracy. There exists a trade-off between the approximation error and the PE level with respect to the separation distance of hidden nodes of RBFNNs.
\end{lemma}



\begin{remark}
 Note that the  optimized  distribution of hidden nodes is calculated using the desired state trajectory as the computation input, without considering the features of the outputs of target functions.  
 Thus, the   optimized distribution is only optimal to the inputs of target functions, rather than being entirely optimal to target functions. 
   Since the outputs of target functions are unknown  a priori, it is hard to consider the features of the outputs  in the structural design of RBFNNs.
\end{remark}

	From the view of the  PE condition, the optimized distribution of the hidden node has three advantages: 
	1) for a specific approximation accuracy, compared with the lattice distribution of hidden nodes scheme, the optimized distribution of the hidden nodes scheme needs less amount of hidden nodes, and the separation distance of hidden nodes is larger. The PE level of the optimized distribution of the hidden nodes scheme is higher than the corresponding of the lattice distribution of the hidden nodes scheme;
	 2) during each period, all hidden nodes experience the similar highest activation value along the desired state trajectory in sequence; 
	 3) for each of the hidden nodes, the number of times when it experiences the highest activation value is similar. 
	 The above   intuitive reasons explain why the proposed scheme   consistently has a higher PE level.
	 Correspondingly, in composite adaptive RBFNN control with a lattice distribution of hidden nodes, the hidden nodes satisfying a partial PE condition have the following behaviours:
	1)  only  a part of  the hidden nodes experience the similar highest activation value periodically along the desired state trajectory in sequence. 
	These hidden nodes which never experience the periodically highest activation value make small contributions to the approximation task. Moreover, these hidden nodes may degrade the robustness of the controller; 
	  2) the number of the highest activation of each hidden node  is different in each period because the distribution of hidden nodes is a lattice which is not matching the shape of state trajectories.

\begin{remark}
	 In terms of  how many hidden nodes we should select for the K-means algorithm, by referring to Lemma \ref{PE_level}, for a specific trajectory, we know that less hidden nodes would lead to 	  a larger separation distance and a higher PE level of RBFNNs but result in a bad approximation accuracy.  	A higher PE level will lead to a higher convergence rate of the adaptive controller. In addition, more hidden nodes would increase the computation costs. There is a trade-off between 	 the approximation error, computation costs, and PE level regarding the number of hidden nodes of the RBFNN. In an actual application, the number of hidden nodes should be decided according to the simulation or experiment phenomenons. If all the weights converge to some constants and the tracking errors converge to zeros,  the number of hidden nodes is enough for the closed-loop system.
\end{remark}

\begin{figure}[!t]
\centering 
\includegraphics[width=3 in]{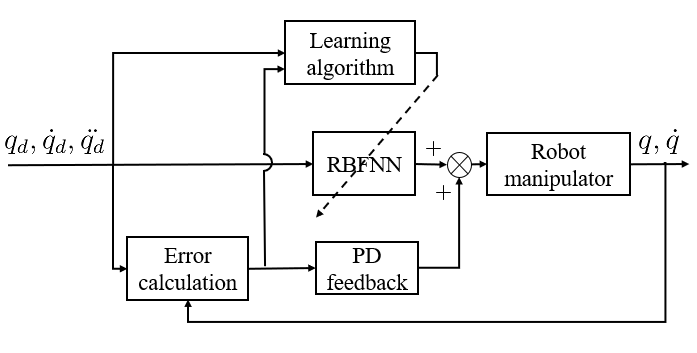}
\caption{The adaptive  feedforward RBFNN control strategy.}
\label{control_diagram}
\end{figure}
    
\section{Control Design}\label{control_design}

 Define the tracking errors as follows:
\begin{equation}
\begin{split}
 e_1&=q_d-q\\
 \dot{e}_1&=\dot{q}_d-\dot{q},
\end{split}
\end{equation}
where $q_d$, $\dot{q}_d$ are the desired joint position and the desired joint velocity.

 A composite tracking error is introduced as:   
  \begin{equation}
  e_2=\dot{e}_1+K_1 e_1,
  \end{equation}
 where $K_1=	  \rm{diag} (K_{11},K_{12},\dots ,K_{1n})$ is 	  a diagonal positive definite matrix.

The composite referenced joint velocity $\dot{q}_r$ and the composite referenced joint acceleration $\ddot{q}_r$ are defined as:
\begin{equation}
\begin{split}
\dot{q}_r&=\dot{q}+e_2=\dot{q}_d+K_1 e_1\\
\ddot{q}_r&=\ddot{q}+\dot{e}_2=\ddot{q}_d+K_1\dot{e}_1.
\end{split}
\end{equation}

The model-based PD-plus-feedforward controller is:
\begin{equation}
\tau = K_{2} e_2 + M(q_d)\ddot{q}_{d}+C(q_d,\dot{q}_d)\dot{q}_{d}+G(q_d),
\label{tau}
\end{equation}
where $K_2=  \rm{diag} (K_{21},K_{22},\dots ,K_{2n})$ is  a diagonal positive definite matrix.     $K_{2} e_2$ can be reshaped as $K_{2} e_2= K_1 K_2 e + K_2 \dot{e}$, and it is a PD term. 

For an adaptive feedforward RBFNN controller, the RBFNN $W^{*T}  S(Z_d)$ is utilized to approximate the feedforward dynamics as follows:
\begin{equation}
\begin{split}
W^{*T}  S(Z_d)+\epsilon(Z_d)=& M(q_d)\ddot{q}_{d}+C(q_d,\dot{q}_d)\dot{q}_{d}\\
&+ G(q_d),
\end{split}
\end{equation}
where $Z_d=[q_d^\text
{T},\dot{q}_{d}^\text
{T},\ddot{
q}_{d}^\text
{T}]^{\text{T}}$ is the input of the adaptive RBFNN, and $\|\epsilon(Z_d)\|^2 \leq \|\bar{\epsilon}\|^2$.

Then, the controller is formulated as:
\begin{equation}
\tau =  K_{2} e_2 +\hat{W}^T  S(Z_d).
\label{tau2}
\end{equation} 
Fig. \ref{control_diagram} illustrates the adaptive feedforward RBFNN control strategy.

A residual error is introduced by replacing the composite dynamics with the feedforward dynamics:
 \begin{equation} \label{residual_robot}
\begin{split}
\tilde{H}_1=& \left( M(q)\ddot{q}_{r}+C(q,\dot{q})\dot{q}_r+G(q) \right)\\
&- \left(  M(q_d)\ddot{q}_{d}+C(q_d,\dot{q}_d)\dot{q}_{d}+ G(q_d)\right). \\
\end{split}
\end{equation}

Another residual error is introduced for the second step of the stability analysis:
\begin{equation} \label{residual_robot_2}
\begin{split}
\tilde{H}_2 = & \left( M(q)\ddot{q}_{r}+C(q,\dot{q})\dot{q}+G(q) \right)\\
&- \left(  M(q_d)\ddot{q}_{d}+C(q_d,\dot{q}_d)\dot{q}_{d}+ G(q_d)\right). \\
\end{split}
\end{equation}

%

In the same manner as Remark 3 of  \cite{Xian2004ACA}, since $M(q)$, $C(q,\dot{q})$ and $G(q)$ are of class $\mathcal{C}^{1}, \forall \boldsymbol{q}, \dot{\boldsymbol{q}} \in \mathbb{R}^{n}$.  Define $E:=[e_{1}^T, e_{2}^T]^T$.    The Mean Value Theorem can be applied on $\tilde{H}_1$, $\tilde{H}_2$ to obtain:
\begin{equation}\label{errorH}
\begin{split}
\left\|\tilde{H}_1\right\| &\leq \rho_1(\|E\|)\|E\|,\\
\left\|\tilde{H}_2\right\| &\leq \rho_2(\|E\|)\|E\|,
\end{split}
\end{equation}
in which $
\rho_1, \rho_2: \mathbb{R}^{+} \mapsto \mathbb{R}^{+}
$
are certain functions that are globally invertible and strictly increasing\cite{Queiroz1997}.

%

The gradient method with  a switching $\delta$-modification is utilized to train the RBFNN:

\begin{equation} \label{adaWS}
\dot{\hat{W}}_i=\Gamma_i \left(S(Z_d) e_{2i}-\delta \hat{W}_i\right),
\end{equation}
where $\delta$ is:
\begin{equation}
\delta=\left\{\begin{array}{ll}
0 & \text { if }\|\hat{W}_i\|<W_{0} \\
\delta_{0} & \text { if }\|\hat{W}_i\| \geq W_{0},
\end{array}\right.
\end{equation}
$\delta_{0}$ is a small positive constant, and the learning rate $\Gamma$ is a positive constant. $W_{0}$ should be selected large
enough so that    $W_{0} > \|W_i^*\|$ considerably.  
Because   $\|W_i^*\|$ is unknown, we can select an obviously large initial value and then reselect it according to the experiment phenomenon. If $ {W}_{0}$ is set less than $\|W_i^*\|$, it has the same properties as the fixed $\delta$-modification.

\begin{remark}
In the fixed $\delta$-modification scheme \cite{weihe2020,CongWang2006}, $\delta$ is a constant and is used for avoiding  the fact that the weights of RBFNNs drift to  infinity. However, it introduces oscillations to the parameter weights of RBFNNs, which would inhibit the learning process.
	The ideal form of $\delta$-modification is $\delta (W_i^*-\hat{W}_i)$ and the term would drive $\hat{W}_i$ to approach $W_i^*$.   
  Since $W_i^*$ cannot be known  a priori,  to avoid that $\hat{W}_i$ drifts to infinity,   a common practice is to set $\delta$-modification as  $\delta (0-\hat{W}_i)$, where $\delta$-modification will drive $\hat{W}_i$ to approach $0$.
	 Switching $\delta$-modification can reduce the oscillation by setting $\delta=0$, when $\|\hat{W}_i\| \leq {W}_{0}$.   
\end{remark}

\begin{theorem}
For  the robotic manipulator \eqref{system}, under Assumption \ref{periodic}, driven by the controller \eqref{tau} with the learning algorithm \eqref{adaWS},  the   tracking errors $e_1$, $e_2$, and the estimated weights error $\tilde{W}$  will exponentially converge to the corresponding small intervals;   
and the intervals can be arbitrarily  diminished by increasing the control gains $K_1$ and $K_2$ and the learning rate $\Gamma$.
\end{theorem} 
 \begin{proof}
1) Proving the boundedness of tracking errors $e_1$,  $e_2$ and the estimated weights error $\tilde{W}_i$. 
 
Consider the following Lyapunov function:
\begin{equation}
V=\frac{1}{2}e_1^Te_1 + \frac{1}{2}e_2^T M e_2+ \frac{1}{2}\sum_{i=1}^n\tilde{W}_i^T \Gamma^{-1} \tilde{W}_i.
\end{equation}

The derivative of $V$ is:
\begin{equation} \label{dv21}
\begin{split}
\dot{V} =& -e_1^T K_1 e_1 +e_2^Te_1 \\
&+ e_2^T M \dot{e}_2 +\frac{1}{2} e_2^T \dot{M} e_2 + \sum_{i=1}^n \tilde{W}_{i}^T \Gamma^{-1} \dot{\tilde{W}}_{i}\\
=&-e_1^T K_1 e_1 +e_2^Te_1 \\
&+ e_2^T \left(M \dot{e}_2 +C e_2\right) + \sum_{i=1}^n \tilde{W}_{i}^T \Gamma^{-1} \dot{\tilde{W}}_{i}\\
   \end{split}
\end{equation}
where $\tilde{W}_{i}=W_{i}^*-\hat{W}_{i}\label{ew}$.

Let us recast the error function $M \dot{e}_2 +C e_2$.  The RBFNN can be reformulated as:
\begin{equation}\label{WSXF2}
\begin{split}
\hat{W} ^T  S (Z_d)=&M(q)\ddot{q}_r+C(q,\dot{q})\dot{q}_r+G(q)\\
&-\tilde{H}_1 -\epsilon(Z)-\tilde{W}^T  S(Z_d).
\end{split}
\end{equation}

Applying the aforementioned result \eqref{WSXF2} into controller \eqref{tau2}, we have:
\begin{equation}\label{tau22}
\begin{split}
\tau  = & K_2 e_2 + M(q)\ddot{q}_r+C(q,\dot{q})\dot{q}_r+G(q)\\
&-\epsilon(Z)-\tilde{W}^T S(Z_d)-\tilde{H}_1.
\end{split}
\end{equation}

Substituting controller \eqref{tau22} into closed-loop system \eqref{system}, the  error equation of the closed-loop system can be obtained:
\begin{equation}\label{errd2}
\begin{split}
M(q) \dot{e}_2 + C(q,\dot{q}) e_2 =&-K_2 e_2 +\tilde{H}_1+\epsilon(Z)\\
&+\tilde{W}^T  S(Z_d).
\end{split}
\end{equation}

Substituting the error equation \eqref{errd2} into  \eqref{dv21} yields:
\begin{equation}\label{dv41}
\begin{split}
\dot{V}=&-e_1^T K_1 e_1 +e_2^Te \\
&+e_2^T \big(-K_2 e_2+\tilde{H}_1+\epsilon(Z)+\tilde{W}^T  S(Z) \big) \\
&-\sum_{i=1} ^n \tilde{W}_{i}^T \Gamma^{-1} \dot{\hat{W}}_{i}\\
\leq&    - \lambda_{\min}(K_1-\frac{1}{2}) \|e_1\|^2 - \lambda_{\min} (K_2-1) \|e_2\|^2 \\
&+e_2^T \tilde{H}_1  + \frac{1}{2}\bar{\epsilon}^2 + \sum_{i=1}^n \tilde{W}_{i}^T ( S(Z)e_{2i} -\Gamma^{-1} \dot{\hat{W}}_{i}).\\
\end{split}
\end{equation}

 Since $\|e_{2}\| \leq \|E\|$, 
according to  $\tilde{H}_1$ in  \eqref{errorH}, we have:
\begin{equation}\label{inequality_e}
 e_{2}^{T} \tilde{H}_1 \leq\|e_{2}\| \|E\| \rho_1(\|E\|) \leq \rho_1(\|E\|) \|E\|^{2}. 
\end{equation}

 Substituting the inequality \eqref{inequality_e} and the learning algorithm \eqref{adaWS} into \eqref{dv41}, we have:

\begin{equation}\label{dv51}
\begin{split}
\dot{V}
\leq& -\lambda_{\min}(K_1-\frac{1}{2})\|e_1\|^2-\lambda_{\min}(K_2-1)\|e_2\|^2 \\
&+\|e\|^{2} \rho_1(\|E\|)+
\frac{1}{2}\bar{\epsilon}^2  +\sum_{i=1}^n \delta \tilde{W}_{i}^T \hat{W}_{i}\\
\leq & -\left( K_{s}-\rho_1(\|E\|)\right) \|E\|^2  +\frac{1}{2}\bar{\epsilon}^2  +\sum_{i=1}^n \delta \tilde{W}_{i}^T \hat{W}_{i},\\
\end{split}
\end{equation}
with $K_{s}=\lambda_{\min}(\lambda_{\min}(K_1-\frac{1}{2}), \lambda_{\min}(K_2-1))$. 

	For the term $\delta \tilde{W}_{i}^T \hat{W}_{i}$ in \eqref{dv51}, when $\|\hat{W}_{i}\|\leq W_{0}$, $ \delta=0$, we have:
\begin{equation}\label{inequality_W1}
\begin{split}
\delta \tilde{W}_{i}^T \hat{W}_{i} &=0\\
&\leq \frac{1}{2}\delta_{0} (  W_{0}^2 - \|\tilde{W}_{i}\|^2).
\end{split}
\end{equation}	

 When $\|\hat{W}_{i}\| > W_{0}$, $\delta=\delta_{0}$,  we have:
 \begin{equation}\label{inequality_W2}
\begin{split}
 \delta \tilde{W}_{i}^T \hat{W}_{i} &=  \delta_{0} \tilde{W}_{i}^T (W^*- \tilde{W}_{i} )\\
 &=   -\delta_{0} \tilde{W}_{i}^T  \tilde{W}_{i} + \delta_{0} \tilde{W}_{i}^T W^* \\
 &\leq -\frac{\delta_{0}}{2}  \tilde{W}_{i}^T  \tilde{W}_{i}  + \frac{\delta_{0}}{2}  \|W^*\|^2\\
 &\leq -\frac{\delta_{0}}{2}  \| \tilde{W}_{i} \|^2  +\frac{\delta_{0}}{2} W_{0}^2.
\end{split}
\end{equation}	
	

Combining  \eqref{inequality_W1} and \eqref{inequality_W2}, we have:
\begin{equation}\label{inequality_W}
\delta \tilde{W}_{i}^T \hat{W}_{i} \leq \frac{1}{2}\delta_{0} (  W_{0}^2 - \|\tilde{W}_{i}\|^2).
\end{equation}

Substituting \eqref{inequality_W} into \eqref{dv51}, we have:
\begin{equation}\label{dv5}
\begin{split}
\dot{V}
\leq & -\left( K_{s}-\rho_1(\|E\|)\right) \||E||^2   -\sum_{i=1}^n\frac{\delta_{0}}{2}  \| \tilde{W}_{i} \|^2 \\
& + \frac{1}{2}\bar{\epsilon}^2  +  \frac{1}{2} n \delta_{0} W_{0}^2.\\
\end{split}
\end{equation}

The domain of $\|E\|$ for $K_{s}-\rho_1(\|E\|) >0$ is estimated by:
\begin{equation}
\Omega_{er}:= \{E\ | \ ||E||<  \rho_1^{-1} (K_s)\},
\end{equation}
where $\rho_1^{-1}(\cdot)$ is the inverse function of $\rho_1(\cdot)$.

$\dot{V}$ can be formulated as the following  format:
\begin{equation} \label{SGUUB}
\dot{V} \leq - c_1 V + c_2 \ \text{for} \  ||E|| < \rho_1^{-1} (K_s),
\end{equation}
where $c_1=\min \big(\frac{2 \lambda_{\min}( K_{s}-\rho(\|E\|)}{ \lambda_{\max}(1, M)}, \frac{\delta_{0}}{\Gamma^{-1}} \big) $ and $c_2= \frac{1}{2}\bar{\epsilon}^2  +\frac{1}{2} n \delta_{0}W_{0}^2$.
 

Integrating \eqref{SGUUB}, we have:
\begin{equation}
V\leq (V(0)-\frac{c_2}{c_1}) \exp^{-c_1 t} +\frac{c_2}{c_1} \leq V(0)+\frac{c_2}{c_1}. 
\end{equation}

Define $D:=2(V(0)+\frac{c_2}{c_1})$,    tracking errors $\|e_1\|$,  $\|e_2\|$ and  the sum of estimated weights error $\sum_{i=1}^n\tilde{W}_i^T \Gamma^{-1}  \tilde{W}_i$ are bounded by: 

\begin{equation}
\begin{split}
\|e_1\| &\leq \sqrt{D}\\
\|e_2\| &\leq \sqrt{\frac{D}{\lambda_{\min}(M)} }\\
\sum_i^n\tilde{W}_{i=1}^T \Gamma^{-1}  \tilde{W}_i &\leq \sqrt{D}.
\end{split}
\end{equation}

To guarantee the stability of the closed-loop system, $\Omega_{er}$ can be arbitrarily enlarged by increasing control gains $K_1$ and $K_2$ to contain both boundaries of $\|e_1\|$ and $\|e_2\|$.

2) Prove that tracking errors $e_1$, $e_2$, and the estimated weights error $\tilde{W}$ exponentially converge to the corresponding residual intervals under the condition of $S(Z_d)$ satisfying the PE condition.

Substituting the controller $\tau$ \eqref{tau2} and the residual error $\tilde{H}_2$ \eqref{residual_robot_2} into closed-loop system \eqref{system},  the error  dynamics of the closed-loop system can be obtained:
\begin{equation}
M(q) \dot{e}_2  = -K_2 e_2 +\tilde{H}_2+\epsilon(Z)+\tilde{W}^T  S (Z_d).
\end{equation}

The entire closed-loop system can be expressed as: 
%

\begin{align}  
\left[\begin{array}{c}
\dot{ e}_1\\
\dot{e}_2 \\
\bar{\dot{\tilde{W}}}
\end{array}\right]=&\left[\begin{array}{ccc}
-K_1 & 1 & \mathbf{0} \\
 \mathbf{0}    &   -M^{-1}K_2 & M^{-1}\bar{S}(Z_d)^T\\
 \mathbf{0} & -\Gamma \bar{S}(Z_d)  & \mathbf{0} 
\end{array}\right]   \left[\begin{array}{c}
e_1 \\
e_2\\
\bar{\tilde{W}}
\end{array}\right]\nonumber\\
&+\left[\begin{array}{c}
 \mathbf{0}\\
M^{-1}(\tilde{H}_2 + \epsilon)\\
\delta \Gamma \bar{\hat{W}}
\end{array}\right], 
\end{align}
where 
$\bar{S}(Z_d)= \rm{diag}(S(Z_d), S(Z_d), \dots, S(Z_d))\in \mathbb{R}^{mn \times n}$, $\bar{\tilde{W}}= [\tilde{W}_1; \tilde{W}_2; \dots; \tilde{W}_n]\in \mathbb{R}^{mn \times 1}$, and $\bar{\hat{W}}= [\hat{W}_1; \hat{W}_2; \dots; \hat{W}_n]\in \mathbb{R}^{mn \times 1}$.

To simplify  the analysis, the system is recast as:

\begin{equation}\label{System_perturbed}
\begin{split}
\left[\begin{array}{c}
\dot{E}\\
\bar{\dot{\tilde{W}}}
\end{array}\right]=&\left[\begin{array}{ccc}
-A & b M^{-1} \bar{S}(Z_d)^T\\
 -\Gamma \bar{S}(Z_d)b^T  & \mathbf{0}
\end{array}\right]   \left[\begin{array}{c}
E\\
\bar{\tilde{W}}
\end{array}\right]\\
&+\left[\begin{array}{c}
bM^{-1}(\tilde{H}_2 + \epsilon) \\
\delta \Gamma \bar{\hat{W}}
\end{array}\right],\\
\end{split}
\end{equation}
 where 
 \begin{equation}
 A= \left[\begin{array}{cc}
K_1 & -1 \\
0     &   M^{-1}K_2
\end{array}\right],
 \end{equation}
$b=[ \mathbf{0}, \mathbf{1}]^T$. $A$ is a positive definite matrix, and  $(A,b)$ is controllable.  The  system \eqref{System_perturbed} is a perturbed system, in which $M^{-1}\tilde{H}_2$ is a vanishing perturbation, and $M^{-1}\epsilon$ is a  non-vanishing perturbation.
  According to the above analysis, we know  \( \hat{W}_i \) is  bounded, and $\delta=0$ when $\|\hat{W}_i\|\leq W_{0}$.   $\delta \Gamma \bar{\tilde{W}}$ is a switch perturbation and the value switches to zero when $\|\hat{W}_i\|\leq W_{0}$.

 The nominal system of \eqref{System_perturbed} is  
 \begin{equation}\label{nominal_system}
\begin{split}
\left[\begin{array}{c}
\dot{E}\\
\bar{\dot{\tilde{W}}}
\end{array}\right]=&\left[\begin{array}{ccc}
-A & b M^{-1}\bar{S}(Z_d)^T\\
 -\Gamma \bar{S}(Z_d)b^T  & \mathbf{0}
\end{array}\right]   \left[\begin{array}{c}
E\\
\bar{\tilde{W}}
\end{array}\right]\\
\end{split}
\end{equation} 

 According to the analysis in Section $\text { VII }$ of \cite{FarrellA1998}, when $S(Z_d)$ is PE, the nominal system \eqref{nominal_system} is global exponentially stable of $(E, \bar{\tilde{W}})=0$.

Consider a perturbed system with the vanishing perturbation $M^{-1}\tilde{H}_2$ as follows:
\begin{equation}\label{System_perturbed_1}
\begin{split}
\left[\begin{array}{c}
\dot{E}\\
\bar{\dot{\tilde{W}}}
\end{array}\right]=&\left[\begin{array}{ccc}
-A & b M^{-1}\bar{S}(Z_d)^T \\
 -\Gamma \bar{S}(Z_d) b^T & \mathbf{0}
\end{array}\right]   \left[\begin{array}{c}
E\\
\bar{\tilde{W}}
\end{array}\right]\\
&+\left[\begin{array}{c}
b M^{-1}\tilde{H}_2 \\
\textbf{0}
\end{array}\right],\\
\end{split}
\end{equation}   

 For the perturbed system with a vanishing perturbation, $\|\tilde{H}_2\| \leq \lambda_h \|E\|$, $\forall E \in \Omega_{ce}$.
According to Lemma 9.1 in \cite{khalil2002nonlinear}, there exists suitably large parameters $K_1$, $K_2$ and $\Gamma$  for the system to  exponentially converge, and the stability is semi-global. 


	Consider the original system \eqref{System_perturbed} constituted by the perturbed system \eqref{System_perturbed_1}, a non-vanishing perturbation $M^{-1}\epsilon$, and a switching perturbation $\delta \Gamma \bar{\hat{W}}$.
The approximation error of the neural network $\epsilon$ is bounded by $\bar{\epsilon}$, which can be arbitrarily small by   designing the neural network.
When $\|\hat{W}_i\|\geq W_{0}$, $\delta= \delta_0$. $\delta_0 \Gamma \bar{\hat{W}}$ can be  rather small by choosing $\delta_0$ small enough.
Since \eqref{System_perturbed_1} is semiglobal exponential stable, according to Lemma 9.2 in \cite{khalil2002nonlinear}, both $E$ and $\bar{\tilde{W}}$  converge to small neighbourhoods of zeros with the size of the neighbourhoods being determined by $\bar{\epsilon}$ and $\delta_0 \Gamma \bar{\hat{W}}$.
When $\tilde{W}_i$ converge to small neighbourhoods, $\|\hat{W}_i\|\leq W_{0}$, $\delta= 0$. Then  both $E$ and $\bar{\tilde{W}}$ converge to small neighbourhoods of zeros with the size of the neighbourhoods being determined only by $\bar{\epsilon}$. 
	The small neighbourhoods can be arbitrarily diminished to zeros under fine-designed neural networks. Further, the intervals can be arbitrarily diminished by increasing the parameters $K_1$, $K_2$ and the learning rate $\Gamma$. 
\end{proof}

\begin{figure*}
\centering
\subfigure[]{\includegraphics[width=3.3in]{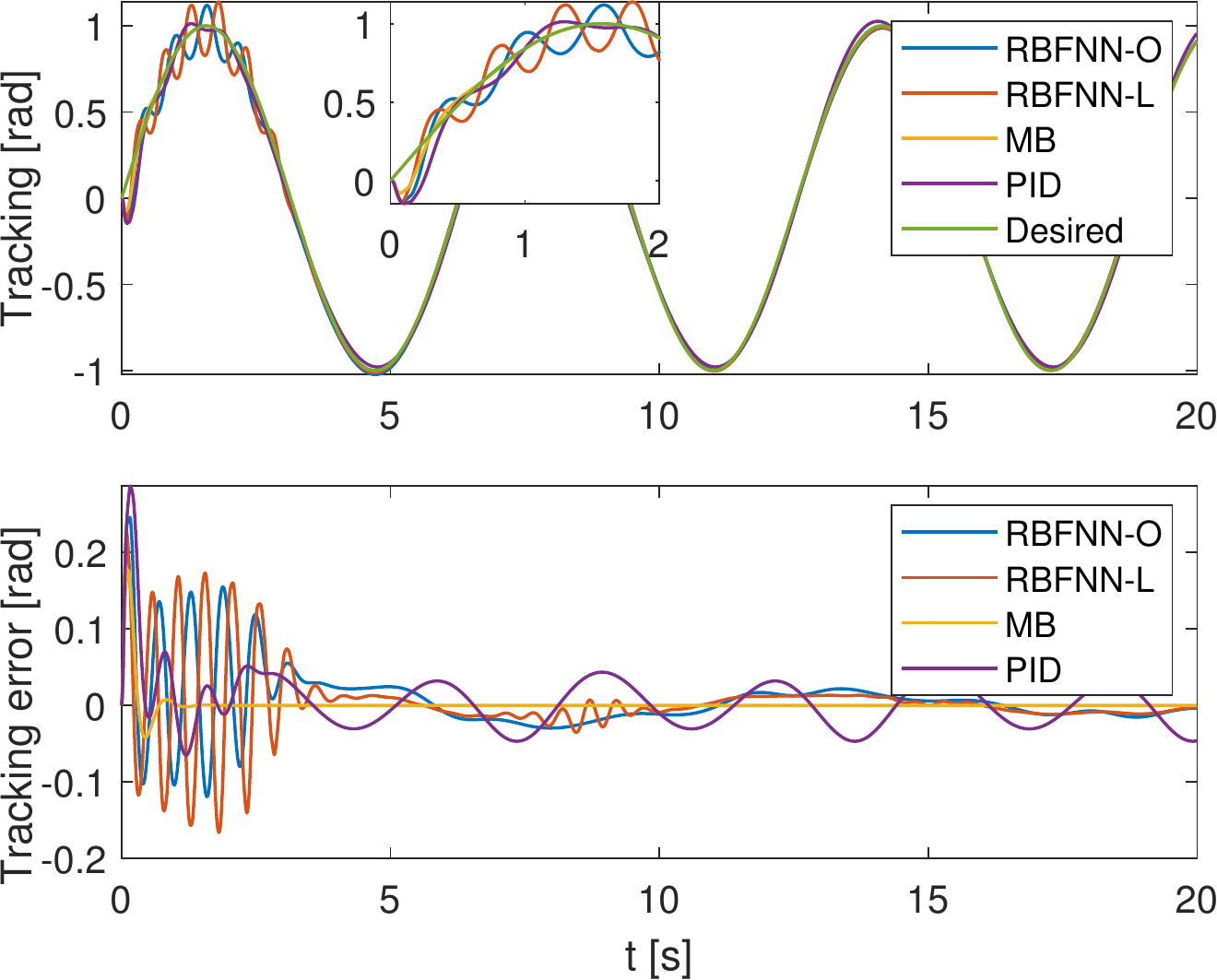}}
\subfigure[]{\includegraphics[width=3.3in]{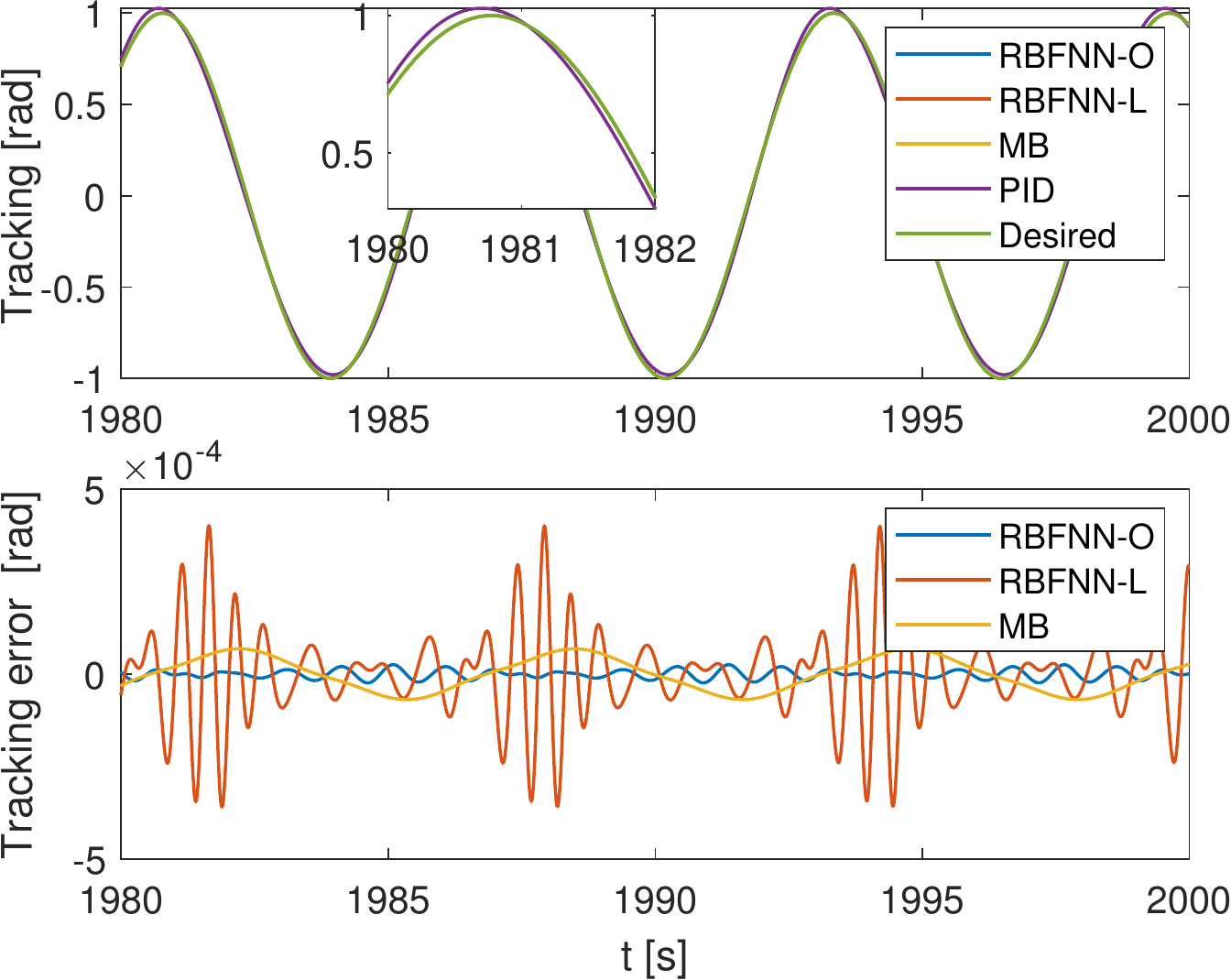}}
\caption{ The tracking performance of Link $1$ by four  controllers.}
\label{Tracking_performance_1}
\end{figure*}

\begin{figure*}
\centering
\subfigure[]{\includegraphics[width=3.3in]{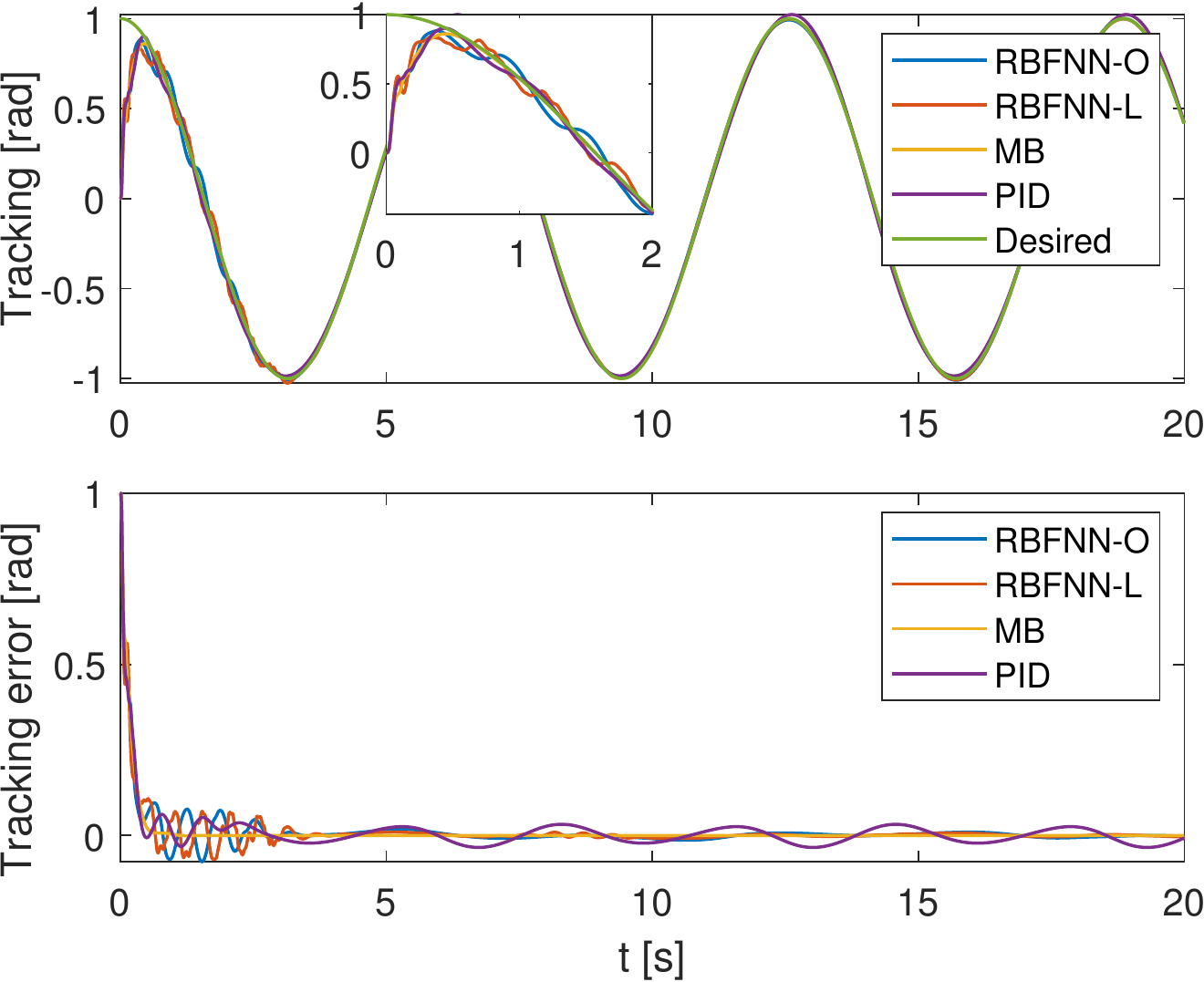}}
\subfigure[]{\includegraphics[width=3.3in]{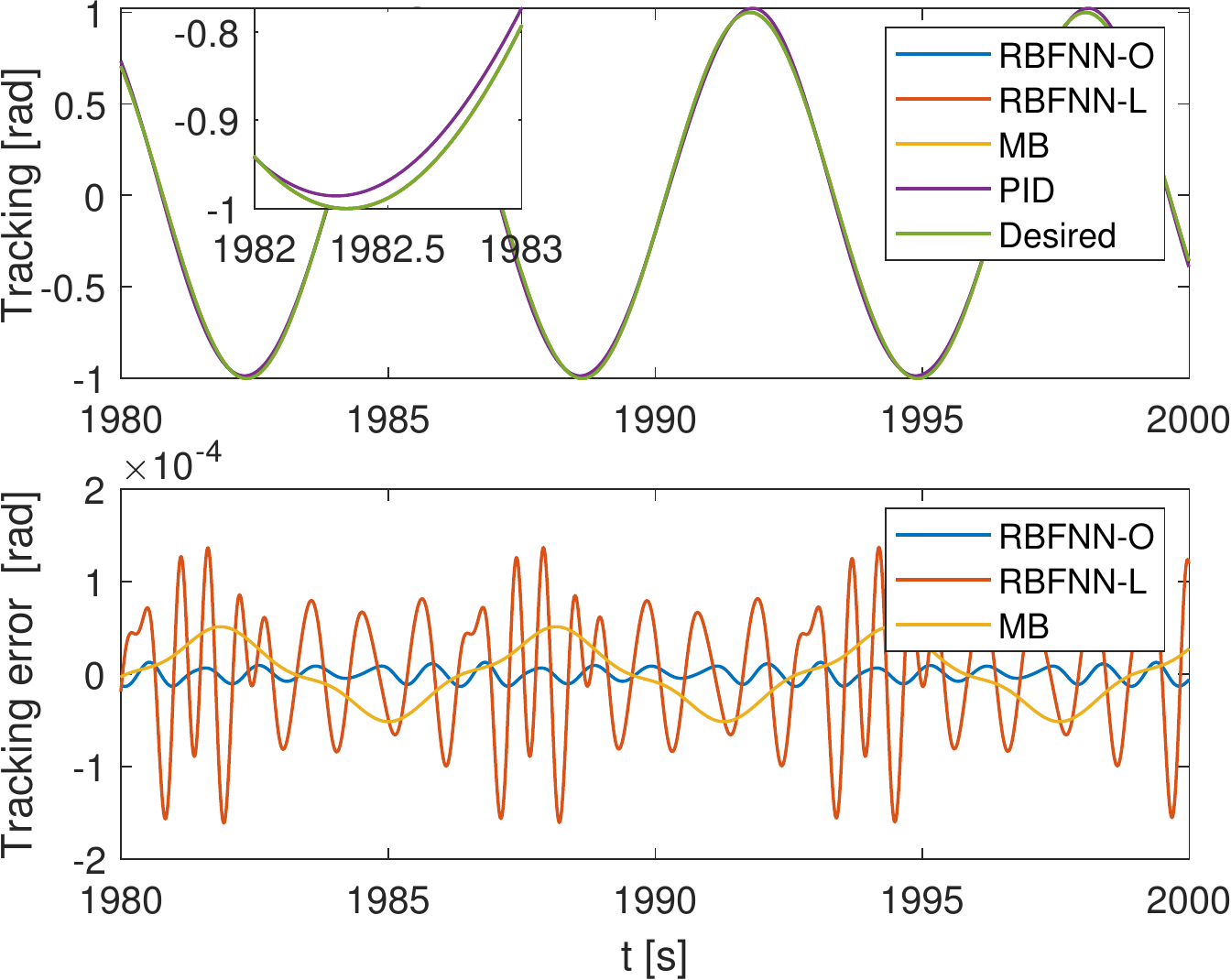}}
\caption{ The tracking performance of Link $2$ by four  controllers.}
\label{Tracking_performance_2}
\end{figure*}

\begin{figure*}
\centering
\subfigure[]{\includegraphics[width=3.3in]{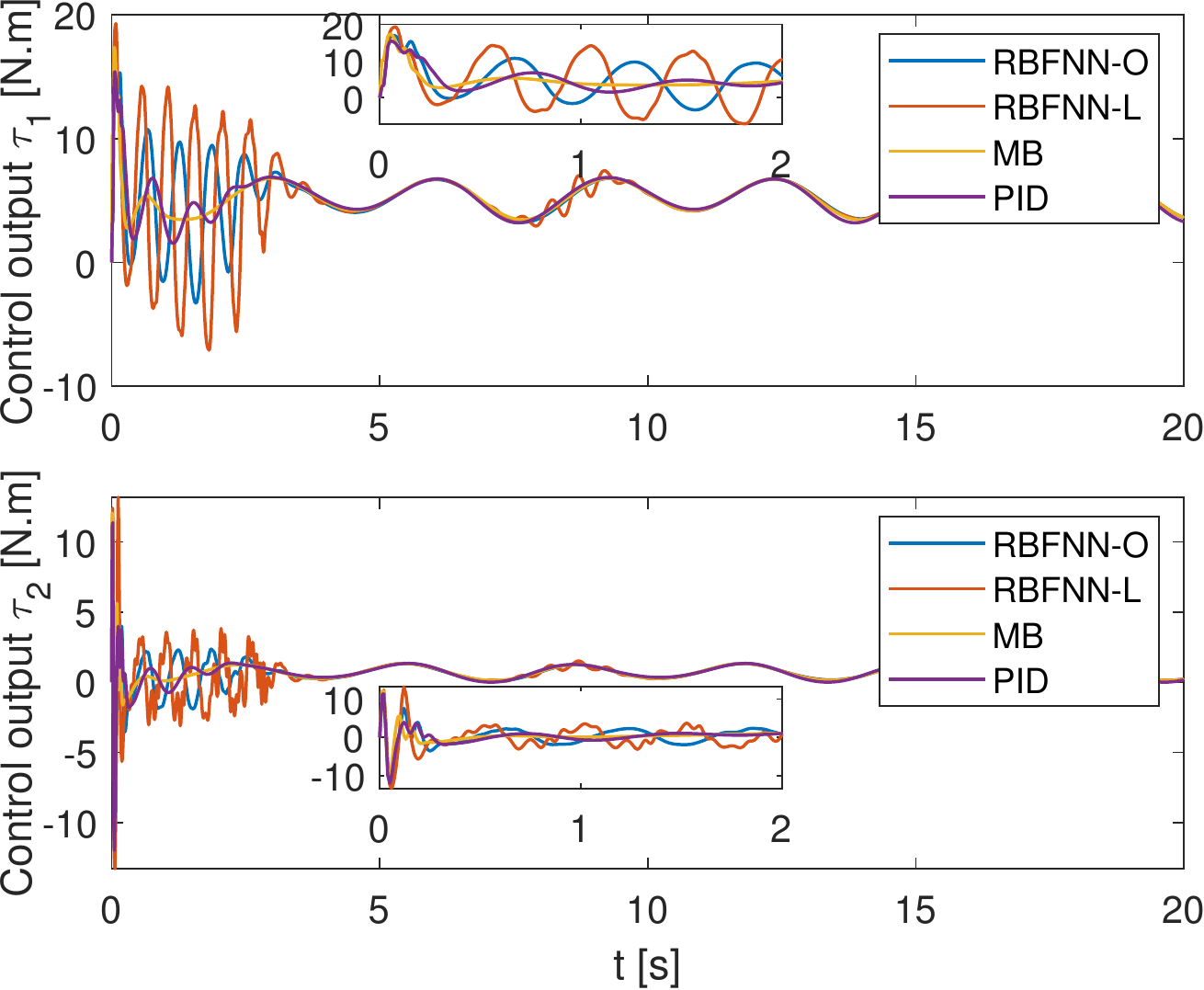}}
\subfigure[]{\includegraphics[width=3.3in]{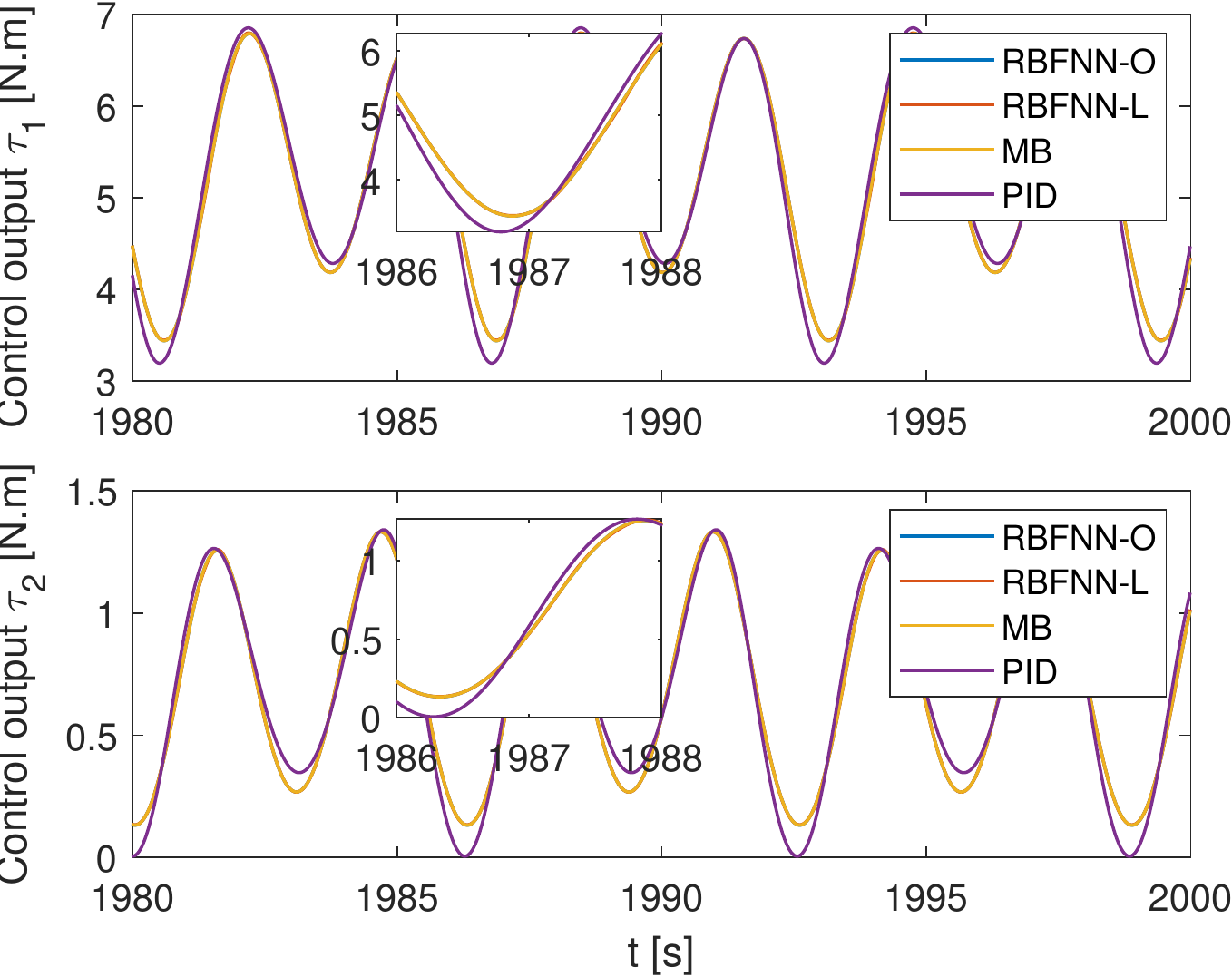}}
\caption{ The control outputs by four controllers.}
\label{control outputs}
\end{figure*}

\begin{figure*}
\centering
\subfigure[]{\includegraphics[width=3.3in]{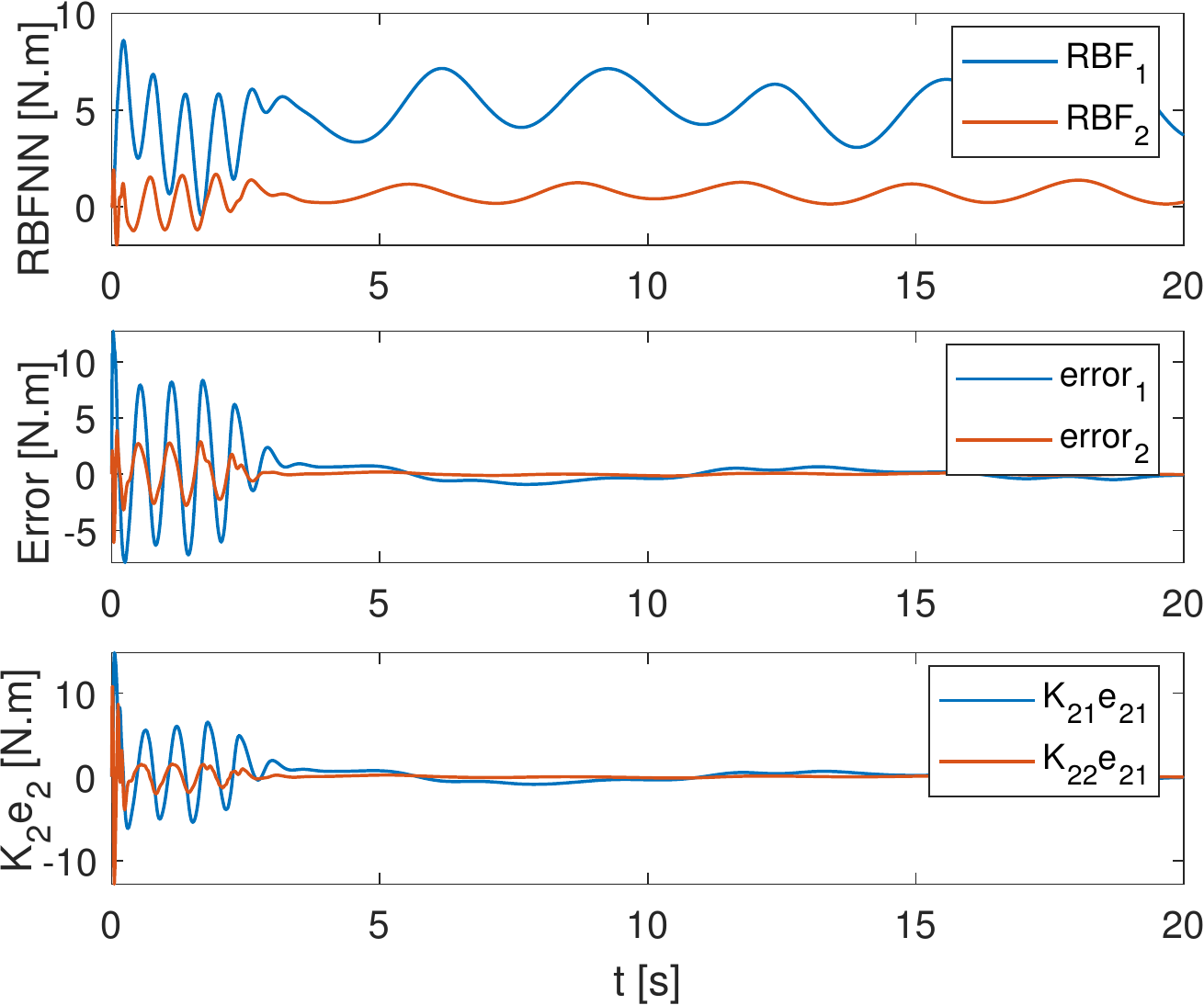}}
\subfigure[]{\includegraphics[width=3.3in]{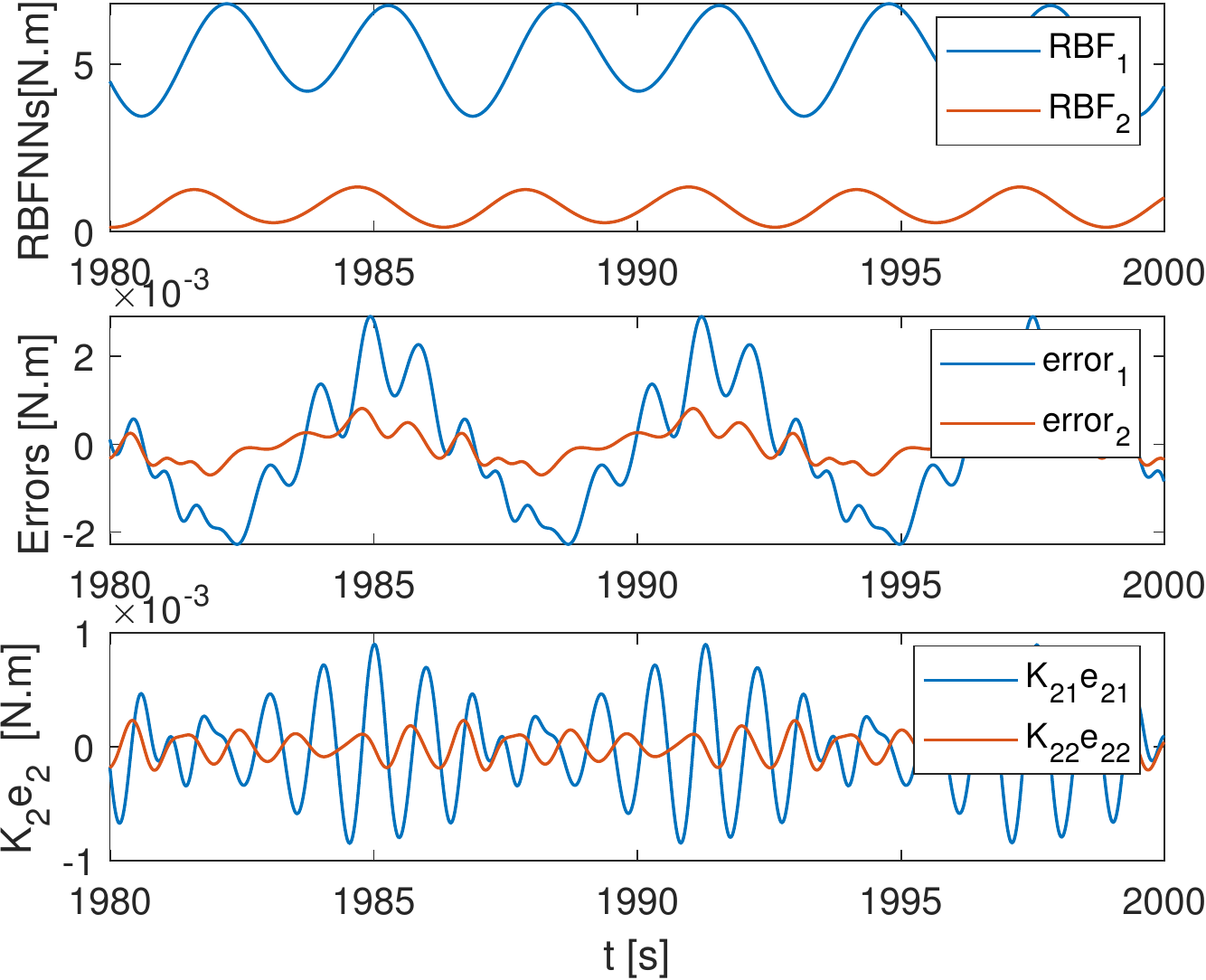}}
\caption{The approximation performance by the  RBFNN-O controller.}
\label{Approximation_performance}
\end{figure*}

\begin{figure*}
\centering
\subfigure[]{\includegraphics[width=3.3in]{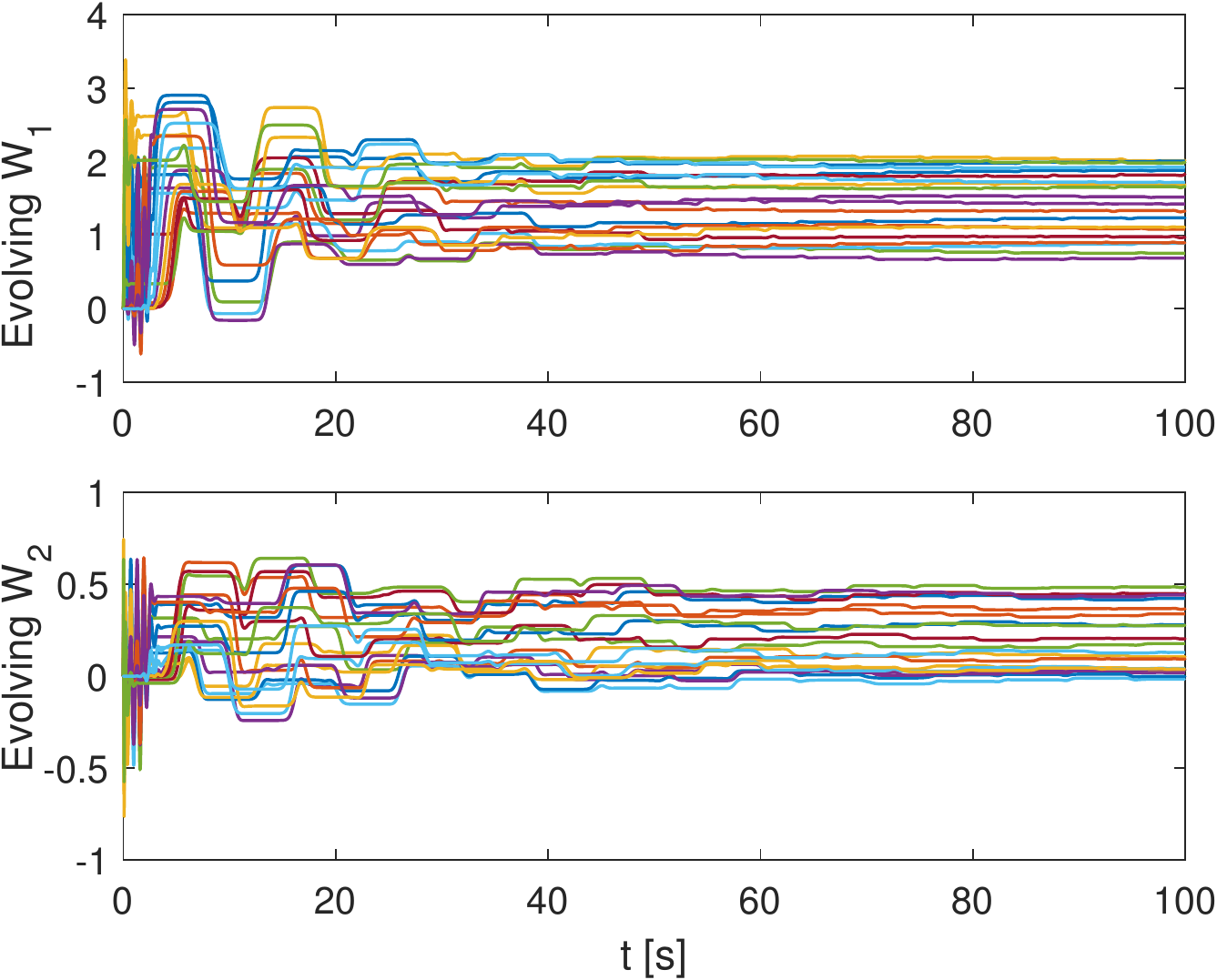}}
\subfigure[]{\includegraphics[width=3.3in]{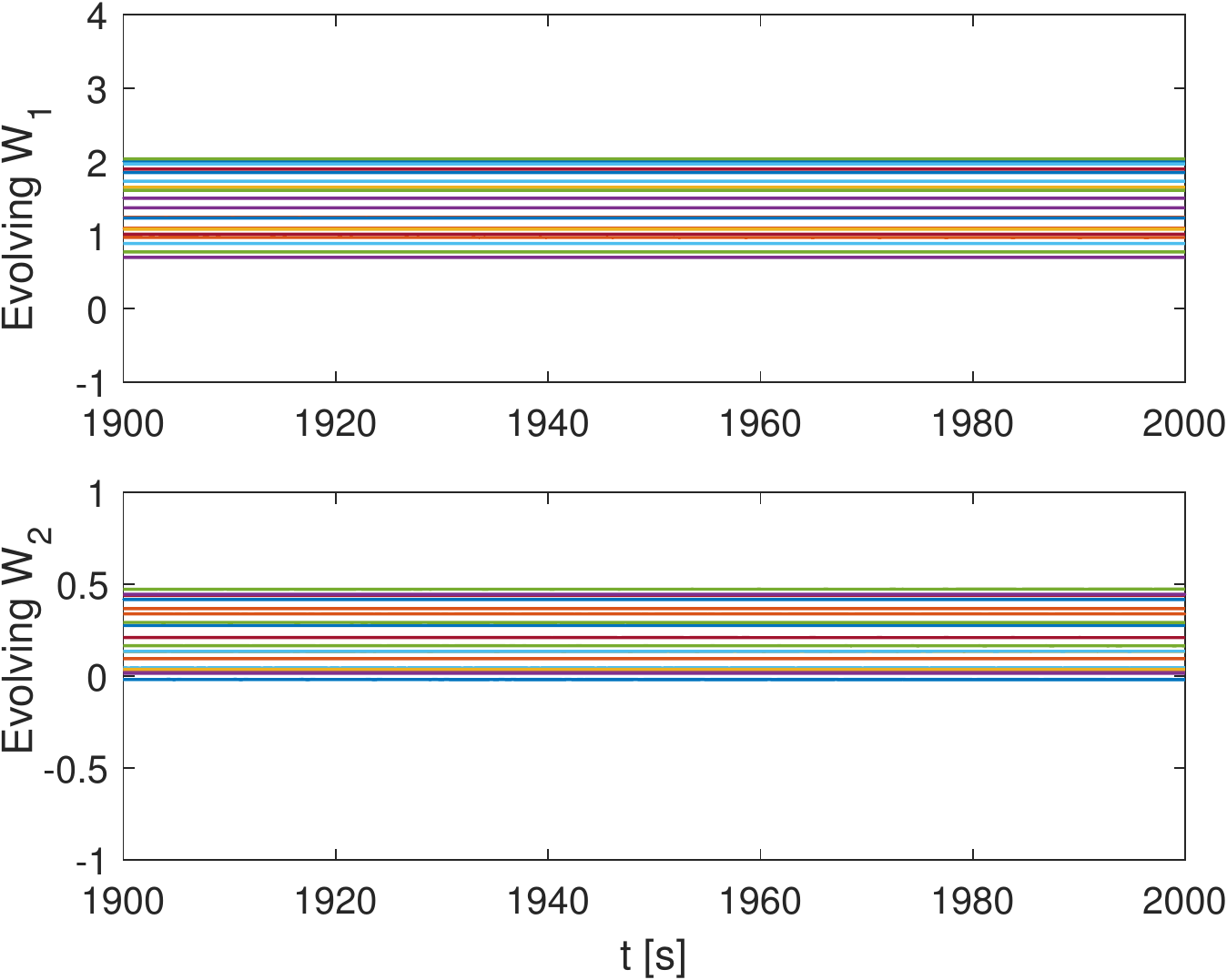}}
\caption{ The learning trajectories of the adaptation weights.}
\label{learning trajectories}
\end{figure*}

\section{Why Can the Proposed Method be Treated as an Enhanced PID  Control?} \label{Discussions}

The proposed control scheme shares a similar rationality to that of the classical PID control in  two special cases, which can thus be seen as an enhanced PID scheme with a better approximation ability.
	To simplify the analysis, we only consider the controller under the circumstance of $\delta=0$ when $\|\hat{W}_{i}\|\leq W_{0}$. 
	 One of the outputs of the adaptive feedforward RBFNN controller \eqref{tau2} can be reshaped to an integral format by substituting learning algorithm \eqref{adaWS}:
\begin{align}
\tau_{i} =&  K_{2i} e_{2i} + \hat{W}_i^T S(Z_d)  \nonumber \\ \label{t_i}
          =& K_{2i} e_{2i} + \Gamma S(Z_d)^T \int S(Z_d)  e_{2i} dt.
\end{align}
	The equation (40) shows that the adaptive feedforward RBFNN control contains a PD term and an integral term. Compared with the integral term in PID  control, the integral term in adaptive feedforward RBFNN control is more complex and has a better approximation capability.

The relations between PID control and adaptive feedforward RBFNN control are comprehended from the following two perspectives:
\begin{itemize}
\item[1)]  When $\sigma \rightarrow \infty$, $S(Z_d)\rightarrow 1_{m\times1}$, the targeted function can be approximated by $F(Z)=W^{*T} 1_{m\times1} +\epsilon(Z)$. The optimal approximated value is $W^{*T} 1_{m\times1} = \frac{\int_{t1}^{t2} F(Z)dt}{t2-t1}$, which means that the degraded RBFNN is only able to approximate constants or horizontal lines. Thus, the controller
 \eqref{t_i} degrades to 
\begin{equation}  
\tau_{i} = K_{2i} e_{2i} + \Gamma m \int e_{2i} dt,
\end{equation} which is  a PID controller.  Decreasing the value of $\sigma$  will enhance the local response of RBFNNs, and thus improve the  approximation accuracy of RBFNNs.
  Then the controller is transformed from a PID  controller to an adaptive feedforward RBFNN controller.

\item[2)]   Consider the simplest circumstance under the desired state $q_d=c_{n\times1}, \;\dot{q}_d=0_{n\times1},\; \ddot{q}_d=0_{n\times1}$, where $(\cdot)_{n\times 1}$ means  an $n\times 1$ dimension vector with a constant value.  There is only one hidden node required to achieve the approximation, and the position of the hidden node is set to be coincident with the desired state position $Z_d=[c_{n\times1}^T,0_{n\times1}^T,0_{n\times1}^T]^T$.
 Under this circumstance,  $S(Z_d)=1$ and the controller degrades to 
\begin{equation}  
\tau_{i} = K_{2i} e_{2i} + \Gamma \int e_{2i} dt,
\end{equation} 
which indicates that the adaptive feedforward RBFNN control for the simplest set point tracking problem is the same as PID  control.  
	We   also get an interesting but a bit weird conclusion that the integral term of a PID  controller satisfies the same PE condition as that of the adaptive feedforward RBFNN controller.  For a PID controller,  $S(Z_d)=1$.  According to definition \ref{definition_PE}, we have  
\begin{equation}  
 \int_{t_{0}}^{t_{0}+T_{0}} S(Z_d) S(Z_d)^{T} d (Z_d) = T_0,
\end{equation} and both $\alpha_{1}$ and $\alpha_{2}$ are equal to $T_0$. 
	 This is also indirectly proved by the exponential stability of the PID controller in \cite{Rocco1996, ALVAREZRAMIREZ200073}.

\end{itemize}

\section{Simulation} \label{Simulation}
Four types of controllers, including the PID  controller, the model-based feedforward (MBFF) controller, the adaptive feedforward RBFNN controller with a lattice distribution of hidden nodes (RBFNN$-$L), and  the adaptive feedforward RBFNN controller with an optimized distribution of hidden nodes (RBFNN$-$O) have been implemented on simulations of a 2-DOF robotic manipulator adopted from Section $3.6.1$ of \cite{geadaptive} to show the superiority of our proposed RBFNN controller.  
 The initial state is $q_1=q_2=0$ and $\dot{q}_1=\dot{q}_2=0$. 
	The desired trajectories are $q_{d1}=\sin(t)$ and $q_{d2}=\cos(t)$. 
 The control gains are  $K_1=[10,0;0,6]$ and $K_2=[3,0;0,1.8]$. The step sizes of those simulations  are $0.01s$.   In both RBFNN-L and RBFNN-O controller, we set the learning rate $\Gamma=6$, the width $\sigma=1.1$, the initial weight $W_1= W_2 = \textbf{0}$, and $W_0=10$.

\begin{itemize}
\item[1)]\textbf{PID}: The control law of the PID controller is $\tau=K_2 r + K_I \int r$, where the control gain $K_I=[0.05,0;0,0.05]$.
\item[2)]\textbf{MBFF}: The control law of the  model-based feedforward controller is $\tau = K_{2} r+ M(q_d)\ddot{q_d}+C(q_d,\dot{q}_d)\dot{q}_d+G(q_d)$, in which the dynamics parameters are accurate.
\item[3)]\textbf{RBFNN-L}: The controller law of the adaptive feedforward RBFNN controller with a lattice distribution of hidden nodes is $\tau =  K_{2} e_2 +\hat{W}^T  S(Z_d)$. $3^6$ hidden nodes are located at  $[-1,0,1]\times[-1,0,1]\times[-1,0,1]\times[-1,0,1]\times[-1,0,1]\times[-1,0,1]$. 
\item[4)]\textbf{RBFNN-O}: The controller law of the adaptive feedforward RBFNN controller with an optimized distribution of hidden nodes is $\tau =  K_{2} e_2 +\hat{W}^T  S(Z_d)$. $20$ hidden nodes are selected.
When calculating the distribution of hidden nodes, the input data was generated from the desired trajectory at the beginning, and then we utilize the K-means algorithm to calculate the distribution of the $20$ hidden nodes. 
\end{itemize}

This paper wants to show  that the proposed controller has better tracking performance than the model-based controller with an accurate dynamics.  The number of hidden nodes is selected as 20  so as to make the RBFNN have better approximation ability. In fact, 10 hidden nodes are enough to obtain decent tracking performance. 
		

  The simulation results are provided as follows.
The	PID, MBFF, and RBFNN-L controllers are selected as the baseline. 
	The tracking performance of the four controllers are as shown in Figs. \ref{Tracking_performance_1} and \ref{Tracking_performance_2}.
The outputs of these controllers  are given in Fig. \ref{control outputs}.
From the simulation results during $0-20s$, we can see that the PID controller only has the primary tracking performance, and the tracking errors of the MBFF controller  converge the fastest. 
	The convergence speeds of both the RBFNN-L and the RBFNN-O controller  are slower than that of the MBFF controller, and the tracking errors of both the RBFNN-L and the RBFNN-O controller are worse than that of the MBFF controller during $0-20s$. 
	However, after enough time to converge, the RBFNN-O controller has smaller tracking errors than the MBFF controller, and the RBFNN-L controller cannot achieve such a perfect performance, which are shown in the simulation results during $1980-2000s$.
	The approximation performance of the adaptive feedforward RBFNN controller is shown in Fig. \ref{Approximation_performance}, from which the approximation errors converge to small intervals.
 $\hat{W}^T  S(Z_d)$ approximates the desired dynamics  such that the proposed controller achieves the accurate tracking performance. It can be verified that the outputs of the RBFNN are about $5 $ and $0.6$   which  account for more than 99\% of the outputs of the controller, whereas the outputs of the PD term are less than $1 \times 10^{-3}$.

	It is weird intuitively that the MBFF controller using precise dynamics has inferior performance in simulation results, whereas the RBFNN-O controller with unknown dynamics generates much better performance.
 In the simulation, the step size is $0.01s$,  which introduces sampling errors.
  The sampling errors in the MBFF controller are not considered, whereas the errors can be well approximated by the adaptive feedforward RBFNN controller.
 This is the main reason why the RBFNN-O controller can prevail over the MBFF controller. 
 
 Intuitively, we use $\hat{W}  S(Z_d)- \big(M(q_d)\ddot{q}_d+C(q_d,\dot{q}_d)\dot{q}_d+G(q_d)\big)$ to represent the approximation errors of the RBFNN. However, the representing form does not consider the sampling errors exiting in the simulation. It is reasonable when the tracking errors are significant because the proportions of the sampling errors in the approximation errors are very small. But the proportions becomes considerable when the tracking errors are getting close to zeros.
	 From Figs. \ref{Tracking_performance_1}, \ref{Tracking_performance_2}, and \ref{Approximation_performance}, when the tracking errors are  less than $2 \times 10^{-4}$, the approximation errors represented by $K_{2} e_2$ are smaller than the ones represented by $\hat{W}^T S(Z_d)- \big(M(q_d)\ddot{q}_d+C(q_d,\dot{q}_d)\dot{q}_d+G(q_d)\big)$. 
	The closed-loop system can be rewritten as $ M(q)\ddot{q}+C(q,\dot{q})\dot{q}+G(q)=K_{2} e_2+ \hat{W}  S(Z_d)$; and the ideal case of the tracking problem is  $K_2 e_2=0$ and $M(q)\ddot{q}+C(q,\dot{q})\dot{q}+G(q)=\hat{W}^T S(Z_d)$. 
	 In reality, this ideal case is almost impossible to achieve. 	
	Hence, during the stable stage, when the tracking errors are small, $K_{2} e_2$ can be used as an indirect index to roughly represent the approximation errors of the closed-loop system.

\begin{table}[!t]
\renewcommand{\arraystretch}{1.3}
\caption{Comparisons of performance indices for three adaptive RNFNNs controllers from $1990s-2000s$}
\label{comparison_simulation}
\centering
\begin{threeparttable}
\begin{tabular}{*{11}{|c|}}
\hline
 \multirow{2}*{Controller}  &  \multicolumn{4}{c|
 |}{Performance indexes} \\
\cline{2-5} 
                & MAAE$_1$\tnote{1}   & MATE$_{1}$\tnote{2} & MAAE$_2$\tnote{3}   & MATE$_{2}$\tnote{4} \\
\hline
PID  & $ 1.32 $ & $0.0432$ &$0.374$ & $ 0.0327$ \\
\hline
MBFF\tnote{5} &  $0.00209$ & $0.0000685$ &   $ 0.000581$   & $0.0000513$\\
\hline
RBFNN-L\tnote{6} & $ 0.0198 $ & $0.000411$ &$0.0033$ & $0.00014$\\
\hline
RBFNN-O\tnote{7} & $ 0.000937 $ & $0.0000267$ &$0.000248$ & $0.0000136$\\
\hline
\end{tabular}
\begin{tablenotes}
\item[1,3] MAAE$_1$ and MAAE$_2$: the maximum absolute approximate error with respect to links 1 and 2, respectively. 
\item[2,4] MATE$_{1}$ and MATE$_{2}$: the maximum absolute tracking error with respect to links 1 and 2, respectively.
\item[5,6,7] MBFF, RBFNN-L, RBFNN-O: the model-based feedforward controller, the adaptive feedforward RBFNN controller with a lattice distribution of hidden nodes, and  the adaptive feedforward RBFNN controller with an optimized distribution of hidden nodes, respectively.
\end{tablenotes}

\end{threeparttable}
\end{table}
  
\begin{remark}
 The MAAE values of  the MBFF controller in Table \ref{comparison_simulation} are not zeros because we utilize $K_2 e_2$ to represent the discrete error of the precise dynamics of the system in the stable stage. 
 Although other engineering tools  can also reduce the discrete errors, this paper applies a relatively simpler method because we aim to show that the adaptive feedforward RBFNN controller can approximate the discrete errors.   
\end{remark}  

The evolutions of the weights $\hat{W}$ are presented in Fig. \eqref{learning trajectories}, which clearly shows that all the weights converge to some constant values after enough time to converge; this also indirectly proves all the hidden nodes satisfy the PE condition.
   For link 1, all values of the weights of the RBFNN are larger than $0.5$, which means that all hidden nodes make significant contributions to the approximation process.
  There are still $5$ weights close to zeros for link 2, whereas this does not means that the corresponding hidden nodes did not satisfy the PE condition. 
  This is because the target function for link $2$ has a small value less than $2N$.
However, for the traditional schemes only satisfying the partial PE condition, most weights of the corresponding hidden nodes, which do not satisfy the PE condition, are  close to zeros.  Only a small part of the weights, of which the corresponding hidden nodes satisfy the PE condition, have considerably high values \cite{pan2016biomimetic, CongWang2006,  Wang2017PElevel}.
 This means that the corresponding hidden nodes not satisfying the PE are not necessary in the approximation process, causing a possible waste of  the computing resource.

	A qualitative comparison of the stable stage performance indices  among four controllers is given in Table \ref{comparison_simulation}. This table shows that both the tracking and approximation performance during the stable stage of the PID controller are the worst among  the four controllers.
	For link $1$, the approximation error and the tracking error of the MBFF controller are $2.2$ times and $2.5$ times larger than those  of the RBFNN-O controller, respectively.
	For link $2$, the approximation error and the tracking error of the MBFF controller are $2.4$ times and $3.7$ times larger than those of the  RBFNN-O controller, respectively.
	The tracking performance of the RBFNN-L controller is   worse than that of both the  RBFNN-O controller and the MBFF controller.

	
\begin{remark}	It is  rather easy to tune our algorithm during the simulation since it is only a bit more complex than a PID controller.
 The tuning method of the control gains of PD terms in the proposed RBFNN scheme is almost the same as that of a PID controller.  
The tuning method for the learning rate in the proposed RBFNN scheme is similar to the integral term of a PID controller. 
\end{remark}  

\section{Conclusion} \label{Conclusion}
This paper has successfully developed an adaptive feedforward RBFNN control strategy for robot manipulators with unknown dynamics. This scheme satisfies the standard PE condition of the RBFNN, while considerably reduces the number of hidden nodes.  	
	The proposed control scheme shares a similar rationality to that of the classic PID control in two special cases,  which can thus be seen as an enhanced PID control with a better approximation ability. 	
	The tuning method of PID can be easily transplanted and applied to our scheme.	
	 Simulation results demonstrate that the tracking performance of our proposed RBFNN controller with unknown dynamics is even better than that of a model-based controller with accurate parameters.
This method could be extended to non-periodic problems by utilizing heuristic algorithms to generate the distribution of hidden nodes  in possible future works.	

\bibliographystyle{IEEEtran}
\bibliography{IEEEabrv,reference}

\begin{thebibliography}{10}
\providecommand{\url}[1]{#1}
\csname url@samestyle\endcsname
\providecommand{\newblock}{\relax}
\providecommand{\bibinfo}[2]{#2}
\providecommand{\BIBentrySTDinterwordspacing}{\spaceskip=0pt\relax}
\providecommand{\BIBentryALTinterwordstretchfactor}{4}
\providecommand{\BIBentryALTinterwordspacing}{\spaceskip=\fontdimen2\font plus
\BIBentryALTinterwordstretchfactor\fontdimen3\font minus
  \fontdimen4\font\relax}
\providecommand{\BIBforeignlanguage}[2]{{%
\expandafter\ifx\csname l@#1\endcsname\relax
\typeout{** WARNING: IEEEtran.bst: No hyphenation pattern has been}%
\typeout{** loaded for the language `#1'. Using the pattern for}%
\typeout{** the default language instead.}%
\else
\language=\csname l@#1\endcsname
\fi
#2}}
\providecommand{\BIBdecl}{\relax}
\BIBdecl

\bibitem{peng2019force}
G.~Peng, C.~Yang, W.~He, and C.~P. Chen, ``Force sensorless admittance control
  with neural learning for robots with actuator saturation,'' \emph{IEEE
  Transactions on Industrial Electronics}, vol.~67, no.~4, pp. 3138--3148, Apr.
  2019.

\bibitem{Hewei2018}
W.~{He} and Y.~{Dong}, ``Adaptive fuzzy neural network control for a
  constrained robot using impedance learning,'' \emph{IEEE Transactions on
  Neural Networks and Learning Systems}, vol.~29, no.~4, pp. 1174--1186, Apr.
  2018.

\bibitem{shi_P_2020_Cybernetics}
J.~{Ni} and P.~{Shi}, ``Adaptive neural network fixed-time leader-follower
  consensus for multiagent systems with constraints and disturbances,''
  \emph{IEEE Transactions on Cybernetics}, pp. 1--14, Feb. 2020.

\bibitem{arabi2019neuroadaptive}
E.~Arabi, T.~Yucelen, B.~C. Gruenwald, M.~Fravolini, S.~Balakrishnan, and N.~T.
  Nguyen, ``A neuroadaptive architecture for model reference control of
  uncertain dynamical systems with performance guarantees,'' \emph{Systems \&
  Control Letters}, vol. 125, pp. 37--44, Mar. 2019.

\bibitem{ge2001stable}
S.~S. Ge, C.~C. Hang, T.~H. Lee, and T.~Zhang, \emph{Stable adaptive neural
  network control}.\hskip 1em plus 0.5em minus 0.4em\relax Springer Science \&
  Business Media, 2002, vol.~13.

\bibitem{CongWang2006}
{Cong Wang} and D.~J. {Hill}, ``Learning from neural control,'' \emph{IEEE
  Transactions on Neural Networks}, vol.~17, no.~1, pp. 130--146, Jan. 2006.

\bibitem{slotine1989composite}
J.~E. Slotine and W.~Li, ``Composite adaptive control of robot manipulators,''
  \emph{Automatica}, vol.~25, no.~4, pp. 509--519, July 1989.

\bibitem{hewei7994622}
W.~{He}, H.~{Huang}, and S.~S. {Ge}, ``Adaptive neural network control of a
  robotic manipulator with time-varying output constraints,'' \emph{IEEE
  Transactions on Cybernetics}, vol.~47, no.~10, pp. 3136--3147, Oct. 2017.

\bibitem{ren2010adaptive}
B.~Ren, S.~S. Ge, K.~P. Tee, and T.~H. Lee, ``Adaptive neural control for
  output feedback nonlinear systems using a barrier lyapunov function,''
  \emph{IEEE Transactions on Neural Networks}, vol.~21, no.~8, pp. 1339--1345,
  July 2010.

\bibitem{gao8879661}
J.~{Qiu}, K.~{Sun}, I.~J. {Rudas}, and H.~{Gao}, ``Command filter-based
  adaptive nn control for mimo nonlinear systems with full-state constraints
  and actuator hysteresis,'' \emph{IEEE Transactions on Cybernetics}, vol.~50,
  no.~7, pp. 2905--2915, July 2020.

\bibitem{yang2018robot}
C.~Yang, C.~Chen, W.~He, R.~Cui, and Z.~Li, ``Robot learning system based on
  adaptive neural control and dynamic movement primitives,'' \emph{IEEE
  transactions on neural networks and learning systems}, vol.~30, no.~3, pp.
  777--787, Mar. 2018.

\bibitem{Sanner1992}
R.~M. {Sanner} and J.~E. {Slotine}, ``Gaussian networks for direct adaptive
  control,'' \emph{IEEE Transactions on Neural Networks}, vol.~3, no.~6, pp.
  837--863, Nov. 1992.

\bibitem{zhao2007locally}
Y.~Zhao and J.~A. Farrell, ``Locally weighted online approximation-based
  control for nonaffine systems,'' \emph{IEEE Transactions on Neural Networks},
  vol.~18, no.~6, pp. 1709--1724, Nov. 2007.

\bibitem{tee2009barrier}
K.~P. Tee, S.~S. Ge, and E.~H. Tay, ``Barrier lyapunov functions for the
  control of output-constrained nonlinear systems,'' \emph{Automatica},
  vol.~45, no.~4, pp. 918--927, Apr. 2009.

\bibitem{liu2016barrier}
Y.~J. Liu and S.~Tong, ``Barrier lyapunov functions-based adaptive control for
  a class of nonlinear pure-feedback systems with full state constraints,''
  \emph{Automatica}, vol.~64, pp. 70--75, Feb. 2016.

\bibitem{Gao8811752}
T.~{Gao}, Y.~{Liu}, D.~{Li}, S.~{Tong}, and T.~{Li}, ``Adaptive neural control
  using tangent time-varying {BLFs} for a class of uncertain stochastic
  nonlinear systems with full state constraints,'' \emph{IEEE Transactions on
  Cybernetics}, pp. 1--11, Aug. 2019.

\bibitem{huang2019motor}
H.~Huang, T.~Zhang, C.~Yang, and C.~P. Chen, ``Motor learning and
  generalization using broad learning adaptive neural control,'' \emph{IEEE
  Transactions on Industrial Electronics}, vol.~67, no.~10, pp. 8608--8617,
  Oct. 2019.

\bibitem{Wang2015PElevel}
M.~{Wang} and C.~{Wang}, ``Learning from adaptive neural dynamic surface
  control of strict-feedback systems,'' \emph{IEEE Transactions on Neural
  Networks and Learning Systems}, vol.~26, no.~6, pp. 1247--1259, June 2015.

\bibitem{Wang2016PElevel}
M.~{Wang}, C.~{Wang}, P.~{Shi}, and X.~{Liu}, ``Dynamic learning from neural
  control for strict-feedback systems with guaranteed predefined performance,''
  \emph{IEEE Transactions on Neural Networks and Learning Systems}, vol.~27,
  no.~12, pp. 2564--2576, Dec. 2016.

\bibitem{Wang2017PElevel}
T.~{Zheng} and C.~{Wang}, ``Relationship between persistent excitation levels
  and {RBF} network structures, with application to performance analysis of
  deterministic learning,'' \emph{IEEE Transactions on Cybernetics}, vol.~47,
  no.~10, pp. 3380--3392, Oct. 2017.

\bibitem{Wang2019_he}
C.~{Yuan}, H.~{He}, and C.~{Wang}, ``Cooperative deterministic learning-based
  formation control for a group of nonlinear uncertain mechanical systems,''
  \emph{IEEE Transactions on Industrial Informatics}, vol.~15, no.~1, pp.
  319--333, Jan. 2019.

\bibitem{chen_wang2019}
T.~{Chen}, D.~J. {Hill}, and C.~{Wang}, ``Distributed fast fault diagnosis for
  multimachine power systems via deterministic learning,'' \emph{IEEE
  Transactions on Industrial Electronics}, vol.~67, no.~5, pp. 4152--4162, May
  2020.

\bibitem{slotine1987on}
J.~E. Slotine and W.~Li, ``On the adaptive control of robot manipulators,''
  \emph{The International Journal of Robotics Research}, vol.~6, no.~3, pp.
  49--59, Sep. 1987.

\bibitem{Chae1987}
{Chae An}, C.~{Atkeson}, J.~{Griffiths}, and J.~{Hollerbach}, ``Experimental
  evaluation of feedforward and computed torque control,'' in
  \emph{Proceedings. 1987 IEEE International Conference on Robotics and
  Automation}, vol.~4, Mar. 1987, pp. 165--168.

\bibitem{khosla1988experimental}
P.~K. Khosla and T.~Kanade, ``Experimental evaluation of nonlinear feedback and
  feedforward control schemes for manipulators,'' \emph{The International
  Journal of Robotics Research}, vol.~7, no.~1, pp. 18--28, Feb. 1988.

\bibitem{reyes2001experimental}
F.~Reyes and R.~Kelly, ``Experimental evaluation of model-based controllers on
  a direct-drive robot arm,'' \emph{Mechatronics}, vol.~11, no.~3, pp.
  267--282, Apr. 2001.

\bibitem{chen2012globally}
W.~Chen, L.~Jiao, and J.~Wu, ``Globally stable adaptive robust tracking control
  using {RBF} neural networks as feedforward compensators,'' \emph{Neural
  Computing and Applications}, vol.~21, no.~2, pp. 351--363, Oct. 2012.

\bibitem{pan2016hybrid}
Y.~Pan, Y.~Liu, B.~Xu, and H.~Yu, ``Hybrid feedback feedforward: An efficient
  design of adaptive neural network control,'' \emph{Neural Networks}, vol.~76,
  pp. 122--134, Apr. 2016.

\bibitem{pan2016biomimetic}
Y.~Pan and H.~Yu, ``Biomimetic hybrid feedback feedforward neural-network
  learning control,'' \emph{IEEE transactions on neural networks and learning
  systems}, vol.~28, no.~6, pp. 1481--1487, Mar. 2016.

\bibitem{sun2018fuzzy}
C.~Sun, H.~Gao, W.~He, and Y.~Yu, ``Fuzzy neural network control of a flexible
  robotic manipulator using assumed mode method,'' \emph{IEEE transactions on
  neural networks and learning systems}, vol.~29, no.~11, pp. 5214--5227, Nov.
  2018.

\bibitem{zeng2014learning}
W.~Zeng and C.~Wang, ``Learning from nn output feedback control of robot
  manipulators,'' \emph{Neurocomputing}, vol. 125, pp. 172--182, Feb. 2014.

\bibitem{ortega2013passivity}
R.~Ortega, J.~A.~L. Perez, P.~J. Nicklasson, and H.~J. Sira-Ramirez,
  \emph{Passivity-based control of Euler-Lagrange systems: mechanical,
  electrical and electromechanical applications}.\hskip 1em plus 0.5em minus
  0.4em\relax Springer Science \& Business Media, 2013.

\bibitem{Kurdila1995}
A.~J. Kurdila, F.~J. Narcowich, and J.~D. Ward, ``Persistency of excitation in
  identification using radial basis function approximants,'' \emph{SIAM J.
  Control Optim.}, vol.~33, no.~2, p. 625–642, Mar. 1995.

\bibitem{lu1998robust}
S.~Lu and T.~Basar, ``Robust nonlinear system identification using
  neural-network models,'' \emph{IEEE Transactions on Neural networks}, vol.~9,
  no.~3, pp. 407--429, May 1998.

\bibitem{zheng2017relationship}
T.~Zheng and C.~Wang, ``Relationship between persistent excitation levels and
  rbf network structures, with application to performance analysis of
  deterministic learning,'' \emph{IEEE transactions on cybernetics}, vol.~47,
  no.~10, pp. 3380--3392, 2017.

\bibitem{Xian2004ACA}
B.~Xian, D.~M. Dawson, M.~S. de~Queiroz, and J.~J. Chen, ``A continuous
  asymptotic tracking control strategy for uncertain nonlinear systems,''
  \emph{IEEE Transactions on Automatic Control}, vol.~49, pp. 1206--1211, July
  2004.

\bibitem{Queiroz1997}
M.~S. {De Queiroz}, {Jun Hu}, D.~M. {Dawson}, T.~{Burg}, and S.~R. {Donepudi},
  ``Adaptive position/force control of robot manipulators without velocity
  measurements: theory and experimentation,'' \emph{IEEE Transactions on
  Systems, Man, and Cybernetics, Part B (Cybernetics)}, vol.~27, no.~5, pp.
  796--809, Sep. 1997.

\bibitem{weihe2020}
W.~{He}, Y.~{Sun}, Z.~{Yan}, C.~{Yang}, Z.~{Li}, and O.~{Kaynak}, ``Disturbance
  observer-based neural network control of cooperative multiple manipulators
  with input saturation,'' \emph{IEEE Transactions on Neural Networks and
  Learning Systems}, vol.~31, no.~5, pp. 1735--1746, May 2020.

\bibitem{FarrellA1998}
J.~A. {Farrell}, ``Stability and approximator convergence in nonparametric
  nonlinear adaptive control,'' \emph{IEEE Transactions on Neural Networks},
  vol.~9, no.~5, pp. 1008--1020, Sep. 1998.

\bibitem{khalil2002nonlinear}
H.~K. Khalil, \emph{Nonlinear Systems}.\hskip 1em plus 0.5em minus 0.4em\relax
  Upper Saddle River, NJ, USA: Prentice Hall, 2002.

\bibitem{Rocco1996}
P.~{Rocco}, ``Stability of {PID} control for industrial robot arms,''
  \emph{IEEE Transactions on Robotics and Automation}, vol.~12, no.~4, pp.
  606--614, Aug. 1996.

\bibitem{ALVAREZRAMIREZ200073}
J.~Alvarez-Ramirez, I.~Cervantes, and R.~Kelly, ``{PID} regulation of robot
  manipulators: stability and performance,'' \emph{Systems \& Control Letters},
  vol.~41, no.~2, pp. 73 -- 83, Oct. 2000.

\bibitem{geadaptive}
S.~S. Ge and C.~J. Harris, \emph{Adaptive Neural Network Control of Robotic
  Manipulators}.\hskip 1em plus 0.5em minus 0.4em\relax World Scientific, 1998.

\end{thebibliography}
\end{document}